\documentclass[sigconf,
]{acmart}
\setcopyright{none}
\renewcommand\footnotetextcopyrightpermission[1]{}
\usepackage[T1]{fontenc}
\usepackage{graphicx}
\usepackage{hyperref}
\usepackage{booktabs}
\usepackage{graphicx}
\usepackage{paralist}
\usepackage{adjustbox}
\usepackage{caption}
\usepackage{float}
\usepackage{subcaption} 
\usepackage{microtype}
\usepackage{amsmath}
\usepackage{multirow} 
\usepackage{listings}
\usepackage{comment}
\usepackage{tabularx}
\usepackage[utf8]{inputenc}
\usepackage{adjustbox}
\usepackage{tikz}
\usepackage{soul}
\usepackage[table]{xcolor}
\usepackage{colortbl}
\usepackage{multicol}
\usetikzlibrary{arrows.meta, positioning, shapes.geometric, fit}
\usetikzlibrary{shapes.geometric}

\usetikzlibrary{positioning, arrows.meta, shapes.misc, fit}
\usepackage[linesnumbered,ruled,vlined]{algorithm2e}
\usepackage[noend]{algpseudocode}

\usetikzlibrary{shapes.multipart, positioning}
\usepackage[utf8]{inputenc}
\usetikzlibrary{calc}
\usepackage{orcidlink}
\usepackage{pgfplots}
\usepackage{pgfplotstable}
\usepackage{pifont}
\usepackage{listings}
\usepackage{xcolor}
\usepackage{enumitem}
\usepackage{stfloats}
\pgfplotsset{compat=1.18}
\usepgfplotslibrary{groupplots}

\usepackage[most]{tcolorbox}

\newcommand{\examplebox}[1]{%
  \begin{tcolorbox}[
    breakable,     
    colback=gray!5!white,
    boxrule=0.1pt,
    left=0.5mm, right=0.5mm, top=0.3mm, bottom=0.3mm
  ]
    {\small#1}
  \end{tcolorbox}%
}

\theoremstyle{definition}
\newtheorem{exmp}{Example}

\definecolor{frstgreen}{RGB}{34,139,34}
\definecolor{orng}{RGB}{255,140,0} 
\definecolor{linkblue}{RGB}{0,102,204}
\newcommand{\myref}[2]{\hyperref[#2]{\textcolor{linkblue}{#1~\ref*{#2}}}}

\newcommand{\filledcircnum}[1]{%
  \tikz[baseline=(char.base)]{
    \node[shape=circle,draw,fill=black,inner sep=1pt] (char) {\color{white}\sffamily\bfseries\tiny #1};}}

\newcommand{\mpara}[1]{\noindent{\bf #1}}

\newcommand{\cmark}{\makebox[\linewidth]{\checkmark}}
\newcommand{\xmark}{\makebox[\linewidth]{$\times$}}

\usepackage[T1]{fontenc}
\usepackage{textcomp}

\usepackage[most]{tcolorbox}
\newcounter{takeawaycounter}
\renewcommand{\thetakeawaycounter}{\arabic{takeawaycounter}}

\usepackage{amsthm}
\theoremstyle{definition}
\newtheorem{definition}{Definition}[section]
\newtheorem{lemma}{Lemma}

\makeatletter
\def\th@definition{%
  \normalfont
  \setlength\parindent{0pt}
  \setlength\parskip{0pt}
}
\makeatother

\tcbset{
  takeawaystyle/.style={
    colback=gray!10,
    colframe=black!70,
    fonttitle=\bfseries,
    boxrule=0.8pt,
    arc=2pt,
    left=4pt,
    right=4pt,
    top=4pt,
    bottom=4pt
  }
}

\tikzset{
  hive/.style={
    regular polygon,
    regular polygon sides=6,
    draw,
    fill=gray!20,
    minimum size=1cm,
    align=center,
    font=\scriptsize
  }
}
\definecolor{honey}{HTML}{ffcb53}

\tikzset{
  hiveperson/.style={hive, fill=honey!100},
  hiveorg/.style={hive, fill=honey!75},
  hivepost/.style={hive, fill=honey!60},
  hiveplace/.style={hive, fill=honey!30}
}

\lstdefinelanguage{XML}
{
  morekeywords={xs:schema, xs:element, xs:complexType, xs:sequence, xs:attribute, xs:annotation, xs:appinfo},
  sensitive=true,
  morecomment=[s]{<!--}{-->},
  morestring=[b]",
}

\lstdefinestyle{myxmlstyle}{
  language=XML,
  basicstyle=\scriptsize\ttfamily,
  backgroundcolor=\color{gray!5},
  showspaces=false,
  showstringspaces=false,
  showtabs=false,
  frame=single,
  captionpos=b,
  breaklines=true,
  breakatwhitespace=true,
  tabsize=1
}
\AtBeginDocument{%
  }

\begin{document}

\title{PG-HIVE: Hybrid Incremental Schema Discovery \\for Property Graphs}

\author{Sophia Sideri}
\email{sophisid@ics.forth.gr}
\orcid{0009-0000-7100-0452}
\affiliation{%
  \institution{ICS, FORTH}
  \city{Heraklion}
  \country{Greece}
}

\author{Georgia Troullinou}
\email{georgia.troullinou@univ-grenoble-alpes.fr}
\orcid{0000-0001-8033-7372}
\affiliation{%
  \institution{CNRS, Univ. Grenoble Alpes}
  \city{Grenoble}
  \country{France}
}

\author{Elisjana Ymeralli}
\email{ymeralli@ics.forth.gr}
\orcid{0009-0009-2748-9815}
\affiliation{%
  \institution{ICS, FORTH}
  \city{Heraklion}
  \country{Greece}
}

\author{Vasilis Efthymiou}
\email{vefthym@hua.gr}
\orcid{0000-0002-0683-030X}
\affiliation{%
  \institution{Harokopio University of Athens}
  \city{Athens}
  \country{Greece}
}

\author{Dimitris Plexousakis}
\email{dp@ics.forth.gr}
\orcid{0000-0002-0863-8266}
\affiliation{%
  \institution{ICS, FORTH}
  \city{Heraklion}
  \country{Greece}
}

\author{Haridimos Kondylakis}
\email{kondylak@ics.forth.gr}
\orcid{0000-0002-9917-4486}
\affiliation{%
  \institution{ICS, FORTH}
  \city{Heraklion}
  \country{Greece}
}

\renewcommand{\shortauthors}{Sideri et al.}
\renewcommand{\shorttitle}{PG-HIVE: Hybrid Incremental Schema Discovery for Property Graphs}

\begin{abstract}
Property graphs have rapidly become the de facto standard for representing and managing complex, interconnected data, powering applications across domains from knowledge graphs to social networks. 
Despite the advantages, their schema-free nature poses major challenges for integration, exploration, visualization, and efficient querying. 
To bridge this gap, we present PG-HIVE, a novel framework for automatic schema discovery in property graphs. PG-HIVE goes beyond existing approaches by uncovering latent node and edge types, inferring property datatypes, constraints, and cardinalities, and doing so even in the absence of explicit labeling information. Leveraging a unique combination of Locality-Sensitive Hashing with property- and label-based clustering, PG-HIVE identifies structural similarities at scale. Moreover, it introduces incremental schema discovery, eliminating costly recomputation as new data arrives. Through extensive experimentation, we demonstrate that PG-HIVE consistently outperforms state-of-the-art solutions, 
in both accuracy (by up to 65\% for nodes and 40\% for edges), and efficiency (up to 1.95x faster execution),
unlocking the full potential of schema-aware property graph management.

\end{abstract}

% \keywords{Schema Discovery, Property Graphs, Graph Databases, Clustering, Locality-Sensitive Hashing}
\acmConference[EDBT'26]{ACM Conference}{XX}{X,
  XX, XX}

\maketitle

\section{Introduction}
Schema discovery is the process of automatically inferring the underlying structure of data without relying on prior schema information~\cite{miller2003schema, paton2024dataset}. 
It has a significant role in modern data management, supporting integration of heterogeneous sources~\cite{DBLP:conf/caise/KondylakisETYP24}, query optimization~\cite{mulder2025optimizing,dumbrava2019approximate}, exploration~\cite{polleres2023does,lissandrini2018x2q}, and data quality assurance~\cite{bonifati2019schema}. 
While schema discovery has been widely studied for relational~\cite{bogatu2020dataset,khatiwada2025tabsketchfm,adelfio2013schema} and semi-structured data~\cite{discala2016automatic,yun2024recg,mior2023jsonoid,dani2024introducing,sadruddin2025llms4schemadiscovery}, property graphs (PGs) pose unique challenges.

Property Graphs are directed multigraphs, where both nodes and edges can have labels and properties (i.e., key-value pairs)~\cite{angles2018property}. 
Unlike relational databases with fixed schemas, graph databases, including property graphs, are more versatile, do not require an explicit schema~\cite{DBLP:journals/corr/abs-2411-09999,DBLP:journals/pvldb/BonifatiDMGJLP22} to be defined, and can dynamically evolve as more data is ingested. Due to this flexibility, property graphs have been widely adopted in various domains, including social networks, healthcare, bioinformatics, and telecommunications~\cite{katal2013big,sahu2020ubiquity}.

Although not requiring a fixed schema offers flexibility, it can also lead to a misconception about the structure of the graph database over time, potentially compromising data integrity and usability~\cite{lissandrini2022knowledge}. To alleviate this problem, schema discovery is crucial for understanding the structure and semantics of property graphs. 
In the context of a property graph, we define its schema as the collection of node and edge types, together with their associated properties, constraints, and cardinalities. A \emph{type} refers to a category of nodes or edges that share a common structural information.

\mpara{The problem.} 
In property graphs, the absence of an explicit schema makes understanding the data’s structure challenging. Nodes and edges can carry arbitrary or missing properties, lack consistent labels, and provide no guarantees of uniformity~\cite{kondylakis2025property}. As graphs evolve, this ambiguity complicates integration, querying, and data quality assurance. Crucially, we cannot assume in advance what constitutes a type or how types should be distinguished, meaning that the data remains opaque until explicitly analyzed. The goal of schema discovery, therefore, is to automatically infer the types of nodes and edges, their properties, and the structural constraints that govern them—while ensuring accuracy, efficiency, and adaptability to dynamic graph scenarios.

\mpara{Challenges.} 
%Both humans and machines make mistakes when generating data, and as this scales, it can pose significant difficulties~\cite{cai2015challenges}. 
%In general, working with big data involves handling a large amount of data, generated at high speed, and diverse in format~\cite{lissandrini2022knowledge,hirzel2018stream,katal2013big}. 
To address the problem of schema discovery, efficient processing methods are needed~\cite{paton2024dataset}. Although schema discovery has been studied for semi-structured RDF data~\cite{kellou2022survey,kardoulakis2021hint}, PGs pose additional challenges~\cite{kondylakis2025property}.
PGs refer to a broad class of enriched graph structures allowing many technical interpretations in different software systems. These interpretations are often incompatible and based on different assumptions~\cite{angles2017foundations}.
Unlike RDF~\cite{antoniou2004semantic,DBLP:conf/er/KondylakisP12,DBLP:conf/sigmod/KondylakisP11}, PGs lack a widely accepted schema language and allow nodes and edges to have arbitrary sets of properties, multiple labels, or none at all~\cite{sakr2021future}. 
Hence, RDF schema discovery methods~\cite{DBLP:conf/bdcsintell/BouhamoumKL18,DBLP:conf/esws/ChevallierKFC24,DBLP:conf/edbt/ChristodoulouPF13,DBLP:conf/er/Kellou-MenouerK15,kellou2022survey,kardoulakis2021hint} cannot be applied directly for inferring the schema of PGs~\cite{ahmetaj2025common}.
Nevertheless, a small number of approaches for PG schema discovery have already emerged - GMMSchema~\cite{DBLP:conf/edbt/BonifatiDM22} and its demo (DiscoPG~\cite{DBLP:journals/pvldb/BonifatiDMGJLP22}), as well as SchemI~\cite{lbath2021schema}. %Unfortunately, they employ 
%clustering \sophia{to schemI den kanei clustering} 
%methods that require \textit{fully labeled data}. 
Although they produce relatively good results under ideal conditions, they employ methods that require \textit{fully labeled data} and they are
less effective in noisy or incomplete datasets (\myref{\S}{sec:evaluation}). 

When data is incomplete, noisy, or integrated from heterogeneous sources, schema discovery becomes a rather difficult problem. 
%OLD
%Clustering is essential for grouping a large amount of data into meaningful types~\cite{zouaq2020schema,DBLP:conf/edbt/BonifatiDM22}. It helps group together nodes and edges with similar properties, labels and semantics. Without clustering, pairwise comparisons between all nodes and edges would be required, which us computationally infeasible for large graphs. 
%Consequently, another challenge is determining the most suitable clustering strategy for a given graph~\cite{}. 
%The effectiveness of schema discovery depends a lot on how well similar structural patterns are grouped together, but each dataset has different characteristics, so one specific method cannot adapt to all cases. This is similar to an approximate nearest neighbor (ANN)~\cite{indyk1998approximate} search, where the optimal grouping depends heavily on the datasets characteristics and the choice of parameter. As a result,
%identifying the appropriate number of clusters, and therefore the correct number of types (\myref{Def.}{def:node_type}, \myref{Def.}{def:edge_type}), remains an open problem in practice, and current techniques offer only approximate guidance~\cite{wang2023graph,wang2024steiner}. \vefthym{not clear}
%NEW
% \HK{replaced the OLD with NEW - check}
Missing properties, inconsistent labels, and conflicting structural patterns obscure the identification of meaningful types. Additionally, we cannot assume that data is missing without domain knowledge, as it can vary in structure. In these cases, clustering plays a central role in grouping nodes and edges based on shared properties, labels, or semantics~\cite{zouaq2020schema,DBLP:conf/edbt/BonifatiDM22}. Without clustering, exhaustive pairwise comparisons would be required, which is computationally infeasible for large graphs. The challenge, however, lies in choosing a clustering strategy that remains effective across diverse datasets. Because each graph exhibits different structural and semantic characteristics, no single method can generalize well. Consequently, identifying the appropriate number of clusters—and thus the correct number of types %(\myref{Def.}{def:node_type}, \myref{Def.}{def:edge_type})
—remains an open problem, with current techniques offering only approximate solutions~\cite{wang2023graph,wang2024steiner}.

\mpara{The solution.} In this paper, we present \textbf{PG-HIVE}, a new open-source approach for schema discovery in large property graph datasets. 
PG-HIVE is designed to infer the schema even in the absence of explicit labels.
It uses both property and label information, when available, in an adaptive manner. To group structurally and semantically similar nodes and edges, clustering is performed using Locality-Sensitive Hashing (LSH)~\cite{datar2004locality,spark_mllib_lsh,DBLP:conf/vldb/GionisIM99,DBLP:conf/cikm/WangCSR13}. The clustering parameters are dynamically adapted to the specific characteristics of each dataset. This forms a \textit{hybrid schema discovery} approach,  as it combines labels, properties, and graph structure to identify types, instead of relying only on the labels or the properties.
% This approach is \textit{hybrid}
% \HK{what is h y b r i d? I have not seen it before in the paper.} 
% by nature, as it leverages labels, properties, and graph characteristics to infer the schema types.

Unlike prior works, which focus only on the type extraction,
PG-HIVE  extracts a detailed schema by identifying \filledcircnum{A} node types, \filledcircnum{B} edge types, \filledcircnum{C} property data types,
\filledcircnum{D} key constraints, and \filledcircnum{E} cardinalities, thus
capturing a richer schema structure. Moreover, the data can be processed incrementally, in batches, allowing the schema to be continuously updated as new data is added without requiring a full recomputation, making PG-HIVE efficient, modular, and suitable for real-world environments.
We provide a thorough experimentation using eight datasets--four real and four synthetic--under varying noise levels, using two LSH-based clustering methods~\cite{leskovec2020mining,DBLP:books/cu/LeskovecRU14}. We explore different configurations, analyzing their impact on schema discovery compared to our adaptive approach. Then, we compare PG-HIVE against state-of-the-art methods (GMMSchema~\cite{DBLP:conf/edbt/BonifatiDM22} and SchemI~\cite{lbath2021schema}), showing that PG-HIVE consistently has better accuracy in schema discovery, and it is particularly superior 
under noisy conditions and when label information is missing, where other approaches degrade. 
Across all datasets, PG-HIVE achieves up to 65\% higher accuracy for nodes, 40\% for edges, and 1.95x faster execution compared to existing methods, demonstrating robustness and efficiency.

To the best of our knowledge, PG-HIVE is the first system that supports an incremental and hybrid schema discovery for property graphs, adapting to the individual datasets' characteristics, and going beyond simple node type discovery. 

% in the absence of label information, where the other approaches degrade.

% , which outperforms them \myref{\S}{sec:evaluation}.

% \HK{this declaration is modest here. Identify the cases that you are far superior, give us one sentence on the superiority - be exact where, and do not say anything about the cases we don't perform good! Focus on the positives}
%and a K-Means~\cite{macqueen1967some} baseline.

\mpara{Outline.} \myref{Section }{sec:related_work} 
discusses related work. \myref{Section }{sec:preliminaries} gives a conceptual description of the problem. \myref{Section }{sec:pg-hive} describes the methodology. \myref{Section }{sec:evaluation} shows our experimental evaluation. \myref{Section }{sec:conclusion} discusses conclusions and future work.

\section{Related Work}
\label{sec:related_work}
% \HK{move this before preliminaries}
% \input{tables/rel_work}

Although schema discovery for big data has been extensively studied 
and explored, especially for RDF~\cite{DBLP:conf/bdcsintell/BouhamoumKL18,DBLP:conf/esws/ChevallierKFC24,DBLP:conf/edbt/ChristodoulouPF13,DBLP:conf/er/Kellou-MenouerK15,kellou2022survey}, tabular~\cite{bogatu2020dataset,khatiwada2025tabsketchfm,adelfio2013schema} and JSON~\cite{discala2016automatic,yun2024recg,mior2023jsonoid,dani2024introducing,sadruddin2025llms4schemadiscovery} data,
only a limited number of approaches~\cite{DBLP:conf/edbt/BonifatiDM22,DBLP:journals/pvldb/BonifatiDMGJLP22,lbath2021schema} have tried to address the complexity and expressiveness of the property graph data model. 

\mpara{RDF approaches.} RDF-based methods cannot be applied directly to PGs due to fundamental differences in the data model~\cite{ahmetaj2025common,angles2018property}. RDF edges are labeled but have no properties, while PG edges may carry properties as key–value pairs. Additionally, PG nodes can have multiple or no labels, indicating their type, while in RDF, types are described by the relation \texttt{rdfs:Class}. Finally, RDF is fundamentally based on ontologies, whereas for PGs, a standard schema definition is not yet available.
% %vasilis: backup from previous version (Yes/No instead of checkmarks) follows in comments
% \begin{table}[t]
% \centering
% \caption{Schema discovery approaches on property graphs.}
% \scriptsize
% \vspace{-3mm}
% \resizebox{\columnwidth}{!}{%
% \begin{tabular}{|p{0.95cm}|p{1.45cm}|p{1.7cm}|p{1.45cm}|p{1.45cm}|}
% \hline
% \textbf{} & \textbf{SchemI \cite{lbath2021schema}} & \textbf{GMMSchema \cite{DBLP:conf/edbt/BonifatiDM22}} & \textbf{DiscoPG \cite{DBLP:journals/pvldb/BonifatiDMGJLP22}} & \cellcolor{honey!50}\textbf{PG-HIVE (ours)} \\
% \hline
% {Label \newline Independent} &  &  &  & \cellcolor{honey!50} \checkmark \\
% \hline
% {Multilabeled \newline Elements} &  & \checkmark  & \checkmark & \cellcolor{honey!50} \checkmark \\
% \hline
% {Schema \newline Elements}   & Nodes \& Edges & Nodes only & Nodes, queries associated Edges & \cellcolor{honey!50}Nodes, Edges \& constraints \\
% \hline
% Constraints &  &  &  & \cellcolor{honey!50} \checkmark \\
% \hline
% Incremental     &  &  & \checkmark & \cellcolor{honey!50} \checkmark  \\
% \hline
% Automation      & \checkmark & \checkmark & \checkmark + UI & \cellcolor{honey!50} \checkmark \\
% \hline
% Notes           & Cannot handle missing labels 
% % Poor performance \myref{\S}{sec:evaluation}. 
% & GMM, cannot handle missing labels & Demo of GMMSchema & \cellcolor{honey!50}LSH and fine tuning \\
% \hline
% \end{tabular}
% }
% \vspace{-5mm}
% \label{tab:rel_work}
% \end{table}
\begin{table}[t]
\centering
\caption{Schema discovery approaches on property graphs.}
\scriptsize
\vspace{-3mm}
\resizebox{\columnwidth}{!}{%
\begin{tabular}{|p{0.95cm}|p{1.45cm}|p{1.7cm}|p{1.45cm}|p{1.45cm}|}
\hline
\textbf{} 
& \textbf{SchemI \cite{lbath2021schema}} 
& \textbf{GMMSchema \cite{DBLP:conf/edbt/BonifatiDM22}} 
& \textbf{DiscoPG \cite{DBLP:journals/pvldb/BonifatiDMGJLP22}} 
& \textbf{PG-HIVE (ours)} \\
\hline

{Label \newline Independent} 
& {\xmark} 
& {\xmark} 
& {\xmark} 
& {\cmark} \\
\hline

{Multilabeled \newline Elements} 
& {\xmark} 
& {\cmark}  
& {\cmark} 
& {\cmark} \\
\hline

{Schema \newline Elements}   
& {\centering Nodes \& Edges} 
& {\centering Nodes only} 
& {\centering Nodes, queries associated Edges} 
& {\centering Nodes, Edges \& constraints} \\
\hline

Constraints 
& {\xmark} 
& {\xmark} 
& {\xmark} 
& {\cmark} \\
\hline

Incremental     
& {\xmark} 
& {\xmark} 
& {\cmark} 
& {\cmark}  \\
\hline

Automation      
& {\cmark} 
& {\cmark} 
& {\cmark} 
& {\cmark} \\
\hline

Notes           
& {\centering Cannot handle missing labels} 
& {\centering GMM, cannot handle missing labels} 
& {\centering Demo of GMMSchema} 
& {\centering LSH and fine tuning} \\
\hline
\end{tabular}
}
\vspace{-5mm}
\label{tab:rel_work}
\end{table}

%\begin{table}[t]
% \centering
% \caption{Schema discovery approaches on property graphs.}
% \scriptsize
% \vspace{-3mm}
% \resizebox{\columnwidth}{!}{%
% \begin{tabular}{|p{0.95cm}|p{1.45cm}|p{1.7cm}|p{1.45cm}|p{1.45cm}|}
% \hline
% \textbf{} & \textbf{SchemI \cite{lbath2021schema}} & \textbf{GMMSchema \cite{DBLP:conf/edbt/BonifatiDM22}} & \textbf{DiscoPG \cite{DBLP:journals/pvldb/BonifatiDMGJLP22}} & \cellcolor{honey!50}\textbf{PG-HIVE (ours)} \\
% \hline
% {Label \newline Independent} & No & No & No & \cellcolor{honey!50}Yes \\
% \hline
% {Multilabeled \newline Elements} & No & Yes  & Yes & \cellcolor{honey!50}Yes \\
% \hline
% {Schema \newline Elements}   & Nodes \& Edges & Nodes only & Nodes, queries associated Edges & \cellcolor{honey!50}Nodes, Edges \& constraints \\
% \hline
% Constraints & No & No & No & \cellcolor{honey!50}Yes \\
% \hline
% Incremental     & No & No & Yes & \cellcolor{honey!50}Yes  \\
% \hline
% Automation      & Yes & Yes & Yes + UI & \cellcolor{honey!50}Yes \\
% \hline
% Notes           & Cannot handle missing labels. 
% % Poor performance \myref{\S}{sec:evaluation}. 
% & GMM, cannot handle missing labels & Demo of GMMSchema & \cellcolor{honey!50}LSH and fine tuning \\
% \hline
% \end{tabular}
% }
% \vspace{-5mm}
% \label{tab:rel_work}
% \end{table}

\mpara{PG schema definitions.} Focusing more on the definition of schema for PGs, this has been defined differently across works, adding complexity to schema discovery. For example, approaches like ~\cite{lbath2021schema} treat each distinct label (e.g., \texttt{Person}, \texttt{Student}) as a separate type, 
%approaches \HK{give me a ref here} or
while several datasets~\cite{himmelstein2017hetionet,takemura2015synaptic,takemura2017connectome,DBLP:conf/sigmod/ErlingALCGPPB15,icij_bahamas} define types by sets of co-occurring labels, e.g., \{\texttt{Student, Person}\} versus \{\texttt{Athlete, Person}\}. 
In integration scenarios, different datasets may use distinct labels for the same conceptual entity (e.g., \texttt{Organization} and \texttt{Company}) or even employ labels in different languages. On the other hand, hierarchical datasets may flatten all nodes under a single generic label (e.g., \texttt{Thing} in CIDOC-CRM~\cite{CIDOC_CRM} datasets). 
These variations lead to ambiguities and inconsistencies, especially when labels are reused across domains (e.g., \texttt{Actor} as a person or as an organization) and complicate schema discovery. 

%Existing approaches often adopt one of these conventions implicitly, which has an impact on the schema discovery results. 
%In our approach we explicitly define types as the PG-Schema model~\cite{DBLP:journals/pacmmod/AnglesBD0GHLLMM23}, where types are formally defined as in Definitions~\ref{def:node_type} and \ref{def:edge_type}, and rely on \emph{patterns} (\myref{Def.}{def:pattern}) to capture variations in properties or labels across instances. 

\mpara{PG schema discovery approaches.} \myref{Table}{tab:rel_work} summarizes the characteristics of the approaches for schema discovery for property graphs.
One of the earliest works, SchemI~\cite{lbath2021schema}, assumes that all nodes and edges are labeled, and groups similar node types based on shared labels. However, this method is limited to datasets with complete type label
declarations, and it cannot infer schemas when labels and properties are missing or inconsistent. A more advanced approach, GMMSchema~\cite{DBLP:conf/edbt/BonifatiDM22}, introduces hierarchical clustering based on Gaussian Mixture Models (GMM)~\cite{bishop2006pattern} to group nodes by analyzing labels and property distributions. However, it has several limitations: (i) it focuses only on node clustering and does not infer relationships between clusters (i.e., edge types), (ii) it assumes fully labeled datasets, (iii) it is not designed to handle missing values or noisy properties, which are common in practice, 
and (iv) it applies sampling techniques to improve performance on large graphs, which impacts the completeness or precision of the inferred schema. 
Additionally, the demo of GMMSchema,  DiscoPG~\cite{DBLP:journals/pvldb/BonifatiDMGJLP22}, incorporates an incremental approach that uses memorization to avoid unnecessary search for types that have already been found. However, it remains fundamentally based on the GMMSchema and queries the associated edges of the discovered nodes, and might lose information of updated clusters.
% , as more data are added.  
%In contrast, our method does not require the presence of schemas in all nodes and edges and is able to effectively process noise and incomplete data. 
% As shown in \myref{Figure}{fig:signif-f1} and \myref{Figure}{fig:exec-signif} our approaches outperforms both prior works. We need to highlight that GMM only clusters node types and therefore has better execution time (more in \myref{Section}{sec:evaluation}).

%This highlights the practical relevance of our method for real-world, noisy and incomplete data.

In summary, most existing approaches either rely on fully annotated datasets, focus on type discovery without modeling constraints, or require manual input. % and sampling. 
They do not support incremental updates or adapt well to noisy and evolving data. \textbf{PG-HIVE} addresses these limitations by providing a hybrid incremental schema discovery framework that models more information about the schema, rather than only the types,
such as constraints and cardinalities, even in the presence of incomplete or evolving data.

%\HK{den mou polyaresei ayto to section apo thm mia giati tha sou poun yparxoun algorithms gia automatic k discovery kai giati exoume kai pollous allous clustering- to shmeiwnw na to skeftw ligo alla mallon kataligw na vgei ektos}
%\mpara{Clustering.} To be able to identify similar nodes and edges, by inferring underlying data patterns, clustering is needed. The most optimal clustering technique, K-Means~\cite{macqueen1967some}, is commonly used but requires predefined cluster numbers, which can be problematic in our approach, as we cannot know beforehand the number of types the dataset holds. 
%In contrast, LSH~\cite{datar2004locality,spark_mllib_lsh,DBLP:conf/vldb/GionisIM99,DBLP:conf/cikm/WangCSR13} offers a more flexible approach, using either Euclidean distance~\cite{DBLP:books/cu/LeskovecRU14}  or MinHash\cite{leskovec2020mining}. LSH enables approximate nearest neighbor searches without needing a fixed number of clusters, making it well-suited for property graphs with high-dimensional, sparse, or noisy data. Other works utilize LSH for schema discovery on data lakes of tabular data~\cite{bogatu2020dataset}, highlighting its effectiveness.

\section{Preliminaries}
\label{sec:preliminaries}
\begin{figure}[t]
\centering
\resizebox{0.9\columnwidth}{!}{%
\begin{tikzpicture}[
  node distance=0.7cm and 1.2cm,
  entity/.style={draw, circle, fill=gray!15, minimum size=0.9cm, align=center, font=\scriptsize},
  person/.style={entity, fill=gray!25},
  org/.style={entity, fill=gray!20},
  post/.style={entity, fill=gray!30},
  place/.style={entity, fill=gray!35},
  literal/.style={draw, dashed, rectangle, font=\tiny, inner sep=1.5pt, align=center, minimum width=1.2cm},
  arrow/.style={->, thin},
  arrow2/.style={-, dashed},
  every node/.style={font=\scriptsize}
]

% Persons
\node[hiveperson] (bob) {Person};
\node[literal, left=0.3 of bob] (b1) {name\\Bob};
\node[literal, above left=0.3cm and 0.4cm of bob] (b2) {gender\\male};
\node[literal, above=0.4cm of bob] (b3) {bday\\2/5/1980};

\node[hiveperson, right=1cm of bob] (alice) {\hypertarget{alice}{}};
\node[literal, above=0.53cm of alice] (a1) {name\\Alice};
\node[literal, below=0.1cm of alice] (a2) {gender\\female};
\node[literal, above right=0.15cm and 0.15cm of alice] (a3) {bday\\19/12/1999};

\node[hiveperson, right=1cm of alice] (john) {Person};
\node[literal, right=0.2cm of john] (j1) {name\\John};
\node[literal, above right=0.3cm and 0.4cm of john] (j2) {gender\\male};
\node[literal, above=0.4cm of john] (j3) {bday\\24/9/2005};

% Posts
\node[hivepost, below left=1.1cm and 0.1cm of bob] (post1) {Post};
\node[literal, above left=0.1cm of post1] (p1) {imgFile\\screenshot.png};

\node[hivepost, below right=1.1cm and 0.1 of john] (post2) {Post};
\node[literal, below=0.2cm of post2] (p2) {content\\bazinga!};

% Organization
\node[hiveorg, right=0.7cm of post1] (org) {Org.};
\node[literal, below=0.2cm of org] (o2) {url\\example.com};
\node[literal, below left=0.35cm and 0.35cm of org] (o1) {name\\Example};

% Place
\node[hiveplace, right=1.5cm of org] (place) {Place};
\node[literal, below=0.22cm of place] (pl1) {name\\Greece};

% Edges
\draw[arrow] (alice) -- node[above, font=\tiny, align=center] {KNOWS} (john);
\draw[arrow] (alice) -- node[above, font=\tiny, align=center] {KNOWS\\\textit{since: 2025}} (bob);
\draw[arrow] (bob) -- node[font=\tiny, align=center,fill=white] {LIKES} (post1);
\draw[arrow] (john) -- node[font=\tiny, align=center,fill=white] {LIKES} (post2);
\draw[arrow] (bob) -- node[ font=\tiny, align=center,fill=white] {WORKS\_AT\\\textit{from: 2000}} (org);
\draw[arrow] (org) -- node[above, font=\tiny, align=center] {LOCATED\_IN} (place);
\draw[arrow] (john) -- node[ font=\tiny, fill=white, inner sep=1pt, align=center] {LOCATED\_IN\\\textit{from: 2025}} (place);

% Properties
\foreach \n in {b1,b2,b3} \draw[arrow2] (\n) -- (bob);
\foreach \n in {a1,a2,a3} \draw[arrow2] (\n) -- (alice);
\foreach \n in {j1,j2,j3} \draw[arrow2] (\n) -- (john);
\foreach \n in {o1,o2} \draw[arrow2] (\n) -- (org);
\draw[arrow2] (pl1) -- (place);
\draw[arrow2] (p1) -- (post1);
\draw[arrow2] (p2) -- (post2);

\end{tikzpicture}
}
% \vspace{-2mm}
\caption{Example Property Graph.}
\label{fig:pg-acm}
\vspace{-4mm}

\end{figure}
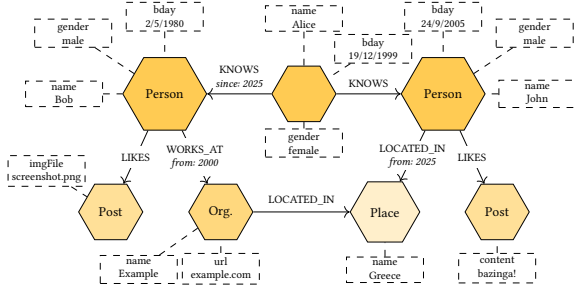

%PGs are a flexible data model commonly used to represent complex structured network data~\cite{angles2018property}. 
A property graph consists of nodes, edges, and associated properties stored as key-value pairs. Nodes represent entities or objects (e.g., \texttt{people, places}), edges define relationships between nodes (e.g.,  \texttt{is\_located\_in, works\_at}), and properties characterize both nodes and edges (e.g., \texttt{name, age, date}).

Although the concept of PGs was first introduced by Rodrigues and Neubauer~\cite{rodriguez2010constructions}, several works have attempted to provide either informal (e.g.,~\cite{Neo4jPlatform,oraclePropertyGraph,DBLP:journals/corr/abs-1909-04881}) or formal (e.g., \cite{DBLP:conf/icde/AlotaibiLQEO21,angles2018property,DBLP:journals/csur/AnglesABHRV17,DBLP:journals/pacmmod/AnglesBD0GHLLMM23,bonifati2018querying,hartig2014reconciliation,DBLP:conf/grades/HartigH19,angles2017foundations,DBLP:journals/ase/SharmaS22}) definitions of the property graph data model. 
%In our approach, we focus on the schema extraction aspect of the property graph model, which tries to define a structural description (e.g., types, properties, relationships, constraints) over the data.  

\mpara{Theoretical concepts.}
In this work, we employ the formal PG-Schema model~\cite{DBLP:journals/pacmmod/AnglesBD0GHLLMM23}, which distinguishes between the data graph and its schema graph. 
PG-Schema builds upon the foundations of the property graph data model and PG-Keys~\cite{DBLP:conf/sigmod/AnglesBDFHHLLLM21} by introducing \emph{keys} and their \emph{constraints}. 
%It is important to note that the term \emph{key} has two different usages in the literature. 
%In our formalization, $K$ denotes the set of \emph{property keys}, i.e., the names of attributes such as \texttt{name}, \texttt{gender}, or \texttt{since} (i.e., key-value pairs). 
Henceforth, we use the term \emph{property keys} exclusively for attribute names (e.g.,  \texttt{gender}, \texttt{name}, \texttt{since}), and we refer to PG-Keys elements as \emph{constraints}.\footnote{By contrast, PG-Keys~\cite{DBLP:conf/sigmod/AnglesBDFHHLLLM21} introduces \emph{keys}, which refer to schema level identification and constraints that assign uniqueness, mandatory properties, edges, or paths.}
% Inthttps://www.overleaf.com/project/6895d81d9be816bc70b13830#, Float, String, Boolean, and ID
PG-Schema uses GQL’s~\cite{graphql-spec-2021} predefined data types, but it can definitely be extended to support a wider range of data types like \texttt{STRING}, \texttt{BOOLEAN}, \texttt{INT}, \texttt{DOUBLE}, \texttt{TIMESTAMP}, or \texttt{DATE,} and can annotate the properties as \texttt{MANDATORY} or \texttt{OPTIONAL} to capture completeness.
% We extend this expressivity by introducing \texttt{BOOLEAN}, \texttt{DOUBLE}, and \texttt{TIMESTAMP}.
In our context, assuming $K$ is a finite set of property keys, $P$ a set of property values, and $\mathcal{DT}$ a set of the aforementioned data types, we define a property graph, and an node/edge type as follows:

\begin{definition}[Property Graph~\cite{DBLP:journals/pacmmod/AnglesBD0GHLLMM23}]\label{def:pg} 
A property graph is a tuple $G = (V, E, \rho, \lambda, \pi),$ where
\begin{itemize}[noitemsep, topsep=0pt, left=0pt]
    \item $V$ and $E$ are disjoint finite sets of nodes and edges, respectively,
    \item $\rho : E \xrightarrow{} (V \times V)$ is a total function mapping edges to ordered pairs of nodes,
    \item $\lambda : (V \cup E) \xrightarrow{} 2 ^ \mathcal{L}$ is a partial function assigning to each node and edge a finite set of labels, 
          % \HK{why do I need this? at this level I don't know anything about types: \st{each label belonging to a type,}} 
          and
        \item $\pi : (V \cup E) \times K \xrightarrow{} P$ is a partial function assigning to each node or edge a finite set of key–value pairs $(k,p)$ with $k \in K$ and $p \in P$.
\end{itemize}
\end{definition}

\begin{definition}[Node Type]\label{def:node_type}
A node type is a tuple $V_s = (\lambda_n, \pi_n)$, where $\lambda_n : V_s \xrightarrow{} 2 ^ \mathcal{L}$ is a finite set of labels, and $\pi_n : V_s \times  \mathcal{K} \to \mathcal{DT} \times \{ \texttt{MANDATORY}, \texttt{OPTIONAL}\}$ assigns to each property key its data type.
\end{definition}

\begin{definition}[Edge Type]\label{def:edge_type}
An edge type is a tuple $E_s = (\lambda_e, \pi_e, \\\rho_e, \mathcal{C}_e)$, where: 
\begin{itemize}[noitemsep, topsep=0pt, left=0pt]
    \item  $\lambda_e : E_s \xrightarrow{} 2 ^ \mathcal{L}$ is a finite set of labels,
    % \HK{similarly I don't need this: \st{associated with the type}},
    \item $\pi_e : E_s \times  \mathcal{K} \to \mathcal{DT} \times \{ \texttt{MANDATORY}, \texttt{OPTIONAL}\}$ assigns to each  property key
    % \HK{what is a property key?:property key} 
    its data type and if its mandatory or optional,
    \item $\rho_e = (V_ss, V_st)$ is the ordered pair of source and target node types associated with the edge type,
    \item \( \mathcal{C} : E_s \to N_1 \times N_2 \) 
    assigns to each edge type a pair of positive integers $N_1,N_2 \in \mathbb{N} = \{0,1,2, \dots \}$, representing the cardinality constraints from source to target; defaults to $N$, if $ N > 1$ .
    % (e.g., \( (1, N) \), \( (N, N) \)).
\end{itemize}
\end{definition}

\noindent Henceforth, we refer to \emph{types} as both nodes and edge types. Then, we assemble the aforementioned and construct the schema as:

% The original PG-Schema model~\cite{DBLP:journals/pacmmod/AnglesBD0GHLLMM23} includes property keys, data types, and cardinalities. Our approach specifically uses \texttt{mandatory}/ \texttt{optional} annotations on properties, in order to capture incompleteness in practical datasets.

% \HK{is mandatory/optional etc. part of the PG-Schema? explain what is part is included there and what it is not in one line}

\begin{definition}[Schema Graph]\label{def:schema_graph}
%Let $\mathcal{{DT}}$ be a set of data types  (e.g., string, boolean, integer, double, date, timestamp).
A schema graph is a tuple  $S_G = (V_s, E_s, \rho_s)$, where:
%\begin{itemize}[noitemsep, topsep=0pt, left=0pt]
    %\item 
    \( V_s \) and \( E_s \) are disjoint finite sets of node types and edge types, and
    % respectively.
    %\item 
    \( \rho_s : E_s \to (V_s \times V_s) \) is a total function mapping edge types to ordered pairs of node types. 
    % \item \( \lambda_s : (V_s \cup E_s) \to T \) is a total function associating a node type labels or edge type labels with a type. 
    % \HK{why I need this? I have the labels inside the node types and edge types. Further, what is "a type"?}
    % \item \( \pi_s : (V_s \cup E_s) \times \mathcal{K} \to \mathcal{DT} \times \{ \texttt{MANDATORY}, \texttt{OPTIONAL} \} \) is a function associating node types or edge types with key-value properties and, for each property, assigning both its data type and whether it is mandatory or optional. \HK{why I need this? I have the property key and their types inside the Node types and Edge types already right?}
    % \item \( \mathcal{C} : E_s \to \mathbb{N} \times \mathbb{N} \) assigns to each edge type a pair of integers representing the cardinality constraints from source to target (e.g., \( (1, N) \), \( (N, N) \)). \HK{similar here. Either you keep C in edge types and remove it from here, or vice versa.}
%\end{itemize}
\end{definition}

% Henceforth, we will refer to $V_s$ and $E_s$ of the schema 
% %graph 
% as \textit{node types} and \textit{edge types}, and the data-level elements simply as \textit{nodes} and \textit{edges}, respectively. \HK{this is already pretty clear from the definitions. I propose  to remove the previous sentence and to clarify in the example 1 which are  nodes and edges and which are node and edge types. Yes you need to show all of them in the example :)}

\examplebox{
\begin{exmp}\myref{Figure}{fig:pg-acm} demonstrates a simple property graph with node types: \texttt{Person}, \texttt{Organization}, \texttt{Post}, \texttt{Place}, and edge types: \texttt{KNOWS}, \texttt{LIKES}, \texttt{WORKS\_AT}, and 
\texttt{LOCATED\_IN}. At the data level, nodes like \textit{Alice}, \textit{Bob}, 
and \textit{John} instantiate these types, while the node \textit{Alice} appears 
without a label, illustrating unlabeled instances. Edges such as 
\texttt{KNOWS}(\textit{Alice}, \textit{John}) or 
\texttt{WORKS\_AT}(\textit{Bob}, \textit{Organization}) represent 
relationships, and attributes (e.g., \texttt{name}, \texttt{gender}, 
\texttt{bday}) appear as dashed literal boxes.

% labeled (\texttt{Person}, \texttt{Organization}, \texttt{Post}, \texttt{Place}) and unlabeled (\textit{Alice node}) node types  and edge types (\texttt{KNOWS}, \texttt{LIKES}, \texttt{WORKS\_AT}, \texttt{LOCATED\_IN}).
\end{exmp}
}

Different works interpret types in PGs differently. Some distinguish types explicitly by labels, the set of labels, or by their semantic meaning. In this paper, we adopt the PG-Schema~\cite{DBLP:journals/pacmmod/AnglesBD0GHLLMM23} model and define the types formally as in \myref{Def.}{def:node_type} and \myref{Def.}{def:edge_type}, allowing different combinations of labels to correspond to different types.
% \HK{
% % other times you use Def. other times Definition. Please be consistent. 
% Now tell us here something like "...in essence allowing different label combinations to lead to different edge and node types"}.
Therefore, we rely on patterns \myref{Def.}{def:node_pattern} and \myref{Def.}{def:edge_pattern} to capture the variety of representations found in the datasets and support schema inference even in the absence of explicit or consistent label annotation. A type might be associated with multiple \emph{patterns}, allowing us to capture and handle noisy or incomplete data.

\begin{definition}[Node Pattern]\label{def:node_pattern} 
Let $\mathcal{L}$ be the set of labels and $\mathcal{K}$ the set of property keys in a property graph. 
A node pattern is a tuple $T_{Np} = (L, K)$, where  $L \subseteq \mathcal{L}$ is a set of labels and $K \subseteq \mathcal{K}$ is a set of property keys. 
% The tuple $T_{Np}$ corresponds to a node pattern, where a node $v$ exists such that $\text{labels}(v) = L$ and $\text{keys}(v) = K$. 
Two node patterns $T_{Np}=(L,K)$ and $T_{Np'}=(L',K')$ are distinct, if $L \neq L'$ and $K \neq K'$. 
\end{definition}

\begin{definition}[Edge Pattern]\label{def:edge_pattern} 
% Let $\mathcal{L}$ be the set of labels and $\mathcal{K}$ the set of property keys in a property graph. 
An edge pattern is a tuple $T_{Ep} = (L, K, R)$, where  $L \subseteq \mathcal{L}$ is a set of labels of an edge, $K \subseteq \mathcal{K}$ is a set of property keys of the edge, and  $R = (L_s, L_t)$ where $L_s, L_t \subseteq \mathcal{L}$ specifies the source and target node label, respectively. 
% node patterns, $R$ is undefined.
Two edge patterns $T_{Ep}=(L,K,R)$ and $T_{Ep'}=(L',K',R')$ are distinct, if $L \neq L'$, $K \neq K'$ and $R \neq R'$. \end{definition}

% \examplebox{
% \begin{exmp}
% In \myref{Figure}{fig:pg-acm} we have several \emph{node patterns}. \textbf{Person} with \texttt{name}, \texttt{gender}, \texttt{bday} (\textit{Bob}, \textit{John}).
% \textbf{Unlabeled} node with \texttt{name}, \texttt{gender}, \texttt{bday} (\textit{Alice}).
% \textbf{Org.} with \texttt{name}, \texttt{url} (\textit{Example}).
% Two patterns for \textbf{Post} one with \texttt{imgFile} (left) and another with \texttt{content} (right).
% \textbf{Place} with \texttt{name} (\textit{Greece}).
  
% For \emph{edge patterns}, we have two \textbf{KNOWS} patterns, one between two \textbf{Person}(\textit{Alice} to \textit{Bob}), with a property \texttt{since} and another pattern without (\textit{Alice} to \textit{John}). \textbf{LIKES} from \texttt{Person} to \texttt{Post} without properties (\textit{Bob} to \textit{screenshot.png} and \textit{John} to \textit{bazinga!}).
%  \textbf{WORKS\_AT} from \texttt{Person} to \texttt{Org.} with a property \texttt{from} (\textit{Bob} to \textit{Example}).
% For \textbf{LOCATED\_IN} we have two patterns, one from \texttt{Org.} to \texttt{Place} without properties (\textit{Example} to \textit{Greece}), and from \texttt{Person} to \texttt{Place} with \texttt{from}(\textit{John} to \textit{Greece}).
% \HK{show me the actual patterns - not in textual description?}

% \end{exmp}
% }
\examplebox{
\begin{exmp}
Based on \myref{Def.}{def:node_pattern},\myref{}{def:edge_pattern}, the patterns of the example in \myref{Figure}{fig:pg-acm} are the following:
\\Node patterns:
\begin{itemize}[noitemsep, topsep=0pt, left=0pt]
    \item $T_{Np1} =(\{\textbf{Person}\}, \{\texttt{name}, \texttt{gender}, \texttt{bday}\})$
    \item $T_{Np2} =(\{\textbf{}\}, \{\texttt{name}, \texttt{gender}, \texttt{bday}\})$
    \item $T_{Np3} =(\{\textbf{Org.}\}, \{\texttt{name}, \texttt{url}\})$
    \item $T_{Np4} =(\{\textbf{Post}\}, \{\texttt{imgFile}\})$,
    \item $T_{Np5} =(\{\textbf{Post}\}, \{\texttt{content}\})$
    \item $T_{Np6} =(\{\textbf{Place}\}, \{\texttt{name}\})$
\end{itemize}
Edge Patterns:
\begin{itemize}[noitemsep, topsep=0pt, left=0pt]
    \item $T_{Ep1} =(\{\textbf{KNOWS}\}, \{\texttt{since}\},(\{\textbf{Person}\}, \{\textbf{Person}\}))$
    \item $T_{Ep2} =(\{\textbf{KNOWS}\}, \{\},(\{\textbf{Person}\}, \{\textbf{Person}\}))$
    \item $T_{Ep3} =(\{\textbf{LIKES}\}, \{\},(\{\textbf{Person}\}, \{\textbf{Post}\}))$
    \item $T_{Ep4} =(\{\textbf{WORKS\_AT}\}, \{\texttt{from}\},(\{\textbf{Person}\}, \{\textbf{Org.}\}))$
    \item $T_{Ep5} =(\{\textbf{LOCATED\_IN}\}, \{\},(\{\textbf{Org.}\}, \{\textbf{Place}\}))$
    \item $T_{Ep6} =(\{\textbf{LOCATED\_IN}\}, \{\texttt{from}\},(\{\textbf{Person}\}, \{\textbf{Place}\}))$
   
\end{itemize}
\end{exmp}
}

% Again, we refer to \emph{patterns} as both node and edge patterns.
We use the term \emph{patterns} to refer to both node and edge patterns.
Multiple patterns can correspond to the same type, 
% (\myref{Def.}{def:node_type},\myref{}{def:edge_type}),
when their associated set of properties differ; in the case of edges, this may also occur when their source or target label sets differ.
% additionally for edges if the source or target label sets differ \HK{rewrite the sentence}.
We refer to the set of all such patterns of a type $\tau \in V_s \cup E_s$ as \emph{type patterns} of $\tau$, formally defined as $P_n= \{(L,K) \mid L = \lambda(\tau)\}$ and $P_e = \{(L,K,R) \;\mid\; L = \lambda(\tau), \; R = (\lambda(s), \lambda(t))\}$.
Thus, two instances with the same labels $L$ but different property sets $K$; and, for edges, the same endpoints $R$, correspond to the same type but belong to different patterns. 
%We refer to each such $(L,K,R) \in \mathcal{P}(\tau)$ as a \emph{type pattern}.

% \HK{
% % I find over complicated the previous definition. What would be the problem if you had two definitions? one for node patterns and one for edge patterns? 
% In any case, you need an example to relate it with node, edge, node type, edge type that is currently missing - continue example 1.}

% \HK{we don't have "types" in this paper. We only have node types and edge types".}
% \HK{Either avoid talking about "types", as in your paper you have node types and edge types only or explain to us that when you speak about types you are refering to node types and edge types.} 

%NEW
\mpara{The schema discovery problem in PGs.} 
\label{sec:schema_issues}
The problem of schema discovery in property graphs can be informally described as:
\textit{``Given a property graph $G$ of arbitrary size and structure, with missing type information, heterogeneous properties, and frequent updates, infer the schema graph $S_G$ efficiently and accurately.''} Several challenges arise in this context:
\begin{itemize}[leftmargin=*,topsep=1pt]
\item \textbf{Label heterogeneity and ambiguity.}
Labels are often used inconsistently~\cite{dallachiesa2019improving,zhou2005learning,ji2011ranking}. Multiple labels may describe the same type (e.g., \texttt{Actor} and \texttt{Person}), or a single label may represent different entities (e.g., \texttt{Actor} as \texttt{Person} or \texttt{Organization}). Structural similarity alone cannot capture such semantic relations (e.g., \texttt{Intern} as a subtype of \texttt{Employee}).

\item \textbf{Efficiency.}
Na\"ive pairwise node comparisons are computationally prohibitive for large graphs.

\item \textbf{Evolving datasets.}
Since real-world graphs evolve, recomputing schemas from scratch is inefficient. Incremental updates are required to maintain performance.

\item \textbf{Schema constraint level.}
Schemas must balance strictness and flexibility~\cite{pandey2025towards}. Strict schemas enforce structure but fail on noisy data, while loose schemas allow variation but reduce precision. PG-Schema~\cite{DBLP:journals/pacmmod/AnglesBD0GHLLMM23} captures this trade-off with \texttt{STRICT} and \texttt{LOOSE} modes.

\end{itemize}

\begin{figure}[h]
\centering
\tikzset{
  bigbox/.style={draw, rounded corners=3pt, align=center, minimum width=30mm, minimum height=11mm, very thick,
                 fill=gray!20},
  smallbox/.style={draw, rounded corners=3pt, align=center, minimum width=30mm, minimum height=9mm, thick, fill=white},
  arrow/.style={-Latex, very thick},
  thinarrow/.style={-Latex, thick},
  dropline/.style={dash pattern=on 2pt off 2pt, thick},
  postframe/.style={draw, rounded corners=4pt},
  labeltitle/.style={font=\bfseries}
}
\resizebox{\columnwidth}{!}{%
\begin{tikzpicture}[node distance=4mm and 4mm]

\node[bigbox]                           (load)    {(\textbf{a}) \textbf{Data Load}\\Nodes, Edges, Properties};
\node[bigbox, right=of load]            (prep)    {(\textbf{b}) \textbf{Preprocess}\\Representation Vectors};
\node[bigbox, right=of prep]            (cluster) {(\textbf{c}) \textbf{Clustering}\\Nodes \& Edges};
\node[bigbox, right=of cluster]         (types)   {(\textbf{d}) \textbf{Type Extraction}\\Merge by labels or\\Jaccard for unlabeled};

\node[smallbox, below=8mm of types, xshift=-87mm] (constraints)
  {(\textbf{e}) \textbf{Property Constraints}\\mandatory / optional};
\node[smallbox, right=3mm of constraints] (dtypes)
  {(\textbf{f}) \textbf{Data Type Inference}\\STRING, INTEGER, DATE, \dots};
\node[smallbox, right=3mm of dtypes] (cards)
  {(\textbf{g}) \textbf{Cardinalities}\\1:1,\;1:N,\;M:N,\dots};

\node[postframe, fit={(constraints) (cards)}, inner sep=16pt] (post) {};
\node[labeltitle] at ($(post.north)+(0,-3mm)$) {Post Processing (optional)};
\node[bigbox, below=5mm of post.south] (schema)
  {(\textbf{h}) \textbf{Schema Serialization}\\PG-Schema \& XSD};

\draw[arrow] (load) -- (prep);
\draw[arrow] (prep) -- (cluster);
\draw[arrow] (cluster) -- (types);

\draw[arrow] (types.east) -- ++(3mm,0) -- ++(0,-19mm) --  (post.east);

\draw[arrow] (post) -- (schema);

\end{tikzpicture}
}
\vspace{-6mm}
\caption{PG-HIVE process.}
\label{fig:workflow}
\vspace{-4mm}
\end{figure}
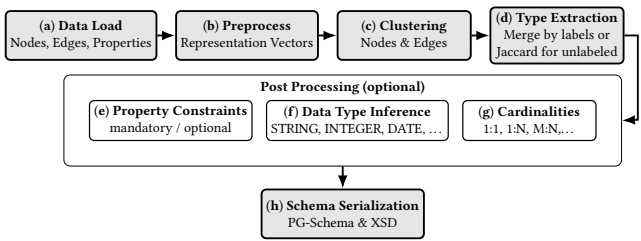

\section{The PG-HIVE Approach}
\label{sec:pg-hive}

To address the aforementioned limitations, we introduce PG-HIVE. 
Its pipeline in \myref{Figure}{fig:workflow} includes: \filledcircnum{a} data loading, \filledcircnum{b} preprocessing, \filledcircnum{c} clustering to group similar nodes and edges, \filledcircnum{d} extracting 
%label types and 
schema types, \filledcircnum{e} inferring property constraints, \filledcircnum{f} property data types, \filledcircnum{g} cardinalities,  and \filledcircnum{h} serializing the schema in XSD~\cite{DBLP:journals/corr/abs-1202-4532} and a PG-Schema \cite{DBLP:journals/pacmmod/AnglesBD0GHLLMM23}. Optionally, this 
% the whole 
pipeline can be executed incrementally. 
% \HK{use same numbering style with figure, either with letters or roman numbers}. 
\vspace{-3mm}

\begin{algorithm}[h]
\small
\SetAlgoLined
\SetKwIF{If}{ElseIf}{Else}
  {\textbf{if}}{\textbf{ then}}
  {\textbf{ else if}}{\textbf{ else}}
  {}

\SetKwFor{ForEach}{\textbf{for each}}{\textbf{ do}}{}
\SetKwFor{For}{\textbf{for}}{\textbf{ do}}{}
\SetKwFor{While}{\textbf{while}}{\textbf{ do}}{}

\caption{PG-HIVE Schema Discovery}
\label{alg:incremental_pghive}

\KwIn{
PG stream $G = \{Gs_{1}, Gs_{2}, \dots, Gs_{n}\}$, 
Jaccard similarity threshold $\theta \in [0,1]$, 
Boolean flag \texttt{postProcessing}
}
\KwOut{Final schema graph $S_G$}

\SetKwFunction{Load}{loadNodesAndEdges}
\SetKwFunction{Pre}{preprocess}
\SetKwFunction{LSH}{LSHClustering}
\SetKwFunction{Merge}{extractTypes}
\SetKwFunction{InferProps}{inferPropertyConstraints}
\SetKwFunction{InferDT}{inferDataTypes}
\SetKwFunction{InferCard}{computeCardinalities}
\SetKwFunction{Update}{updateSchema}

$S_G \gets \emptyset$

\ForEach(){$Gs_{i} \in G$}{
    $D \gets \Load(Gs_{i})$ \\
    $X,b,T \gets \Pre(D)$ 
    
    $\mathcal{C} \gets \LSH(X, b, T)$
    
    $S'_G \gets \Merge(\mathcal{C}, S_G, \theta =0.9)$ \tcp*{\myref{Algorithm}{alg:mergetypes}}
    
    \If{\texttt{postProcessing} \textbf{or} $i = n$}{
        {
        \InferProps{$S'_G$} \\ 
        \InferDT{$S'_G$} \\
        \InferCard{$S'_G$}
        }
    }
    
    $S_G \gets \Update(S'_G)$
}
\Return $S_G$
\end{algorithm}

\vspace{-2mm}

\myref{Algorithm}{alg:incremental_pghive} outlines the process, handling the input graph $G$ in batches ($Gs_1, \ldots, Gs_n$). The main pipeline (Lines 3-6) is executed for each batch, with optional post-processing (Lines 7-9).
The algorithm starts with loading the nodes and edges, and creating the representative vectors (Lines 3-4).
Each node's vector consists of a fixed-dimensional Word2Vec embedding of its labels and a binary vector indicating the presence or absence of each property. Similarly, the edges have three fixed-dimensional Word2Vec embeddings (one for its label and two for the source and target labels), as well as the binary vector indicating its properties. 
The vectors are used then for the LSH clustering (Line 5). An important component is the type extraction step (Line 6), where the types inferred in the current batch are merged and integrated into the final schema $S_G$. 
If a new cluster has the same label as an existing type, it merges with the corresponding type. For unlabeled clusters, we compute the Jaccard similarity between their properties against those of existing types. The remaining unmerged clusters are appended to the new schema as new \texttt{ABSTRACT}~\cite{DBLP:journals/pacmmod/AnglesBD0GHLLMM23} types. Optionally, post-processing steps can be applied after each iteration to infer constraints and cardinalities. In the sequel, we analyze in detail each step.

\subsection{Data Loading \& Representation} 
\label{sec:load}
Firstly, we load the nodes and edges, plus their properties, from a PG storage system (e.g., Neo4j~\cite{Neo4jPlatform})
%instance,  
and transform them into a vector representation (e.g., %Spark 
DataFrame~\cite{armbrust2015spark,zaharia2010spark}), using a single query to ensure similar structure. Let $G$ be our property graph, where  $V = \{v_1,\dots, v_n\}$ denotes the set of nodes and  $E = \{e_1, \dots, e_m\}$ the set of edges in the graph. Each node $v_i$ has a set of properties $\mathcal{K}_{v_i} = \{k_1, \dots, k_k\}$, and each edge $e_j$ a set of properties $\mathcal{K}_{e_j} = \{k_1, \dots, k_l\}$. Edges are also associated with a pair of nodes $(V_s, V_t)$ (\myref{Def.}{def:pg}).

\noindent\textbf{Representation.}
\label{sec:preprocessing}
We transform the data into a machine-readable format (e.g., one-hot encoding) that serves as input to the clustering process.
Each node $v_i$ is described by a vector $\mathbf{f}{v_i} \in \mathbb{R}^{d + K}$, where $d$ is the dimension of the Word2Vec~\cite{mikolov2013efficient} embedding, that will represent the label(s) of each node and $K$ the total number of distinct properties of our dataset. The vector $\mathbf{f}{v_i}$ is the concatenation of the embedding $\mathbf{w}_{v_i} \in \mathbb{R}^d$ and the binary vector $\mathbf{b}_{v_i} \in \{0,1\}^K$ that indicates the presence or absence of each property in each instance. We train a Word2Vec model on the set of node and edge labels observed in the dataset to ensure consistent semantic embeddings across identical label sets.
Similarly, each edge $e_j$ is represented by a vector $\mathbf{f}{e_j} \in \mathbb{R}^{3d + Q}$, composed of three Word2Vec embeddings $\mathbf{w}_{e_j}, \mathbf{w}_{src_j}, \mathbf{w}_{tgt_j} \in \mathbb{R}^d$, for the edge type, the source node type, and the target node type, respectively, and a binary vector $\mathbf{b}_{e_j} \in \{0,1\}^Q$, for the edge properties, where $Q$ is the total number of distinct edge properties of our dataset. 

In the case where the label is absent, we use a zero vector of the same size as the Word2Vec vector (\myref{Example}{ex:representation}). On the other hand, if we have multiple labels on an instance, we sort them alphabetically for uniformity and then concatenate them as one and finally transform them as the rest. 
These hybrid vectors,
combine both the semantic information of the labels and the structural representation for nodes and edges~\cite{bishop2006pattern}. This representation prevents semantically different nodes, or edges, from being merged due to their same structure. 
On the other hand, when an instance has multiple labels, we assume the sorted concatenation of them as a unique label. This helps identify types (nodes or edges) with multiple labels, by having the same Word2Vec if they have multiple same labels or if some of the labels are different.
% This helps in handling  the flattened hierarchical issue \HK{I do not remember the issue here} and identify label sets. 

\examplebox{
\begin{exmp}
\label{ex:representation}
Let's consider the global property keys set 
$\mathcal{K}_n = \{\texttt{imgFile}, \texttt{content}, \texttt{name}, \texttt{url}, \texttt{bday}, \texttt{gender}\}$.
The node ``Bob'' of type \textbf{Person} with properties 
\{\texttt{name}, \texttt{gender}, \texttt{bday}\}, in a Word2Vec size of 5 is represented as the concatenation of:
\vspace{-1mm}
\[ \underbrace{[0.12, 0.85, -0.33, 0.47, 0.19]}_{\text{Word2Vec(``Person'')}} \;\; \Vert \;\; 
\underbrace{[0,0,1,0,1,1]}_{\text{binary props}} \]

For an unlabeled node like ``Alice'', the representation would be like:
\[\underbrace{[0, 0, 0, 0, 0]}_{\text{Word2Vec (unlabeled)}} 
 \;\; \Vert \;\; 
 \underbrace{[1, 0, 1, 0, 1, 1]}_{\text{binary props}}
 \]

Edges have the set of properties $\mathcal{K}_e = \{\texttt{since}, \texttt{from}\}$. We represent ``WORKS\_AT'' from \textbf{Person} to \textbf{Org.} with property \{\texttt{from}\} as: 
\[
\underbrace{[0.44,\,-0.11,\,0.93,\,0.05,\,-0.27]}_{\text{Word2Vec(``WORKS\_AT'')}}
\;\Vert\;
\underbrace{[0.12,\,0.85,\,-0.33,\,0.47,\,0.18]}_{\text{Word2Vec(``Person'')}}
\]
\[
\underbrace{[0.66,\,0.10,\,-0.21,\,0.73,\,-0.08]}_{\text{Word2Vec(``Org.'')}}
\;\Vert\;
\underbrace{[0,1]}_{\text{binary props}}
\]
% In \myref{Figure}{fig:pg-acm}, ``Bob'' is encoded with the Word2Vec embedding for the word ``Person``, and a binary vector for properties. This vector is constructed from the global set of properties observed in the dataset (e.g., \texttt{imgFile}, \texttt{content}, \texttt{name}, \texttt{url}, \texttt{bday}, \texttt{gender}), assigning 1 for present properties and 0 for the rest.
% The edge \texttt{WORKS\_AT},
% uses three embeddings, one for its edge label (``WORKS\_AT''), one for its source label (``Person'') and one for its target label (``Org.''), 
% as well as its corresponding binary vector. Unlabeled nodes (e.g., ``Alice'') use a zero vector embedding, of same size, and their binary vector.
% \HK{can you actually show me an example such vector?}
\end{exmp}
}

\subsection{Clustering}
\label{sec:clustering}
To identify the node and edge types, we employ Locality-Sensitive Hashing (LSH)~\cite{DBLP:books/cu/LeskovecRU14,leskovec2020mining,gionis1999similarity,indyk1998approximate} for clustering their vector representation. LSH has been utilized for several applications such as nearest neighbour search~\cite{ferhatosmanoglu2001approximate,houle2005fast,indyk1998approximate,tao2010efficient}, document clustering~\cite{broder1998min,broder1997resemblance}, and RDF data management~\cite{alucc2019building,kardoulakis2021hint}. In our context, we use LSH to cluster nodes and edges with highly similar vectors,
%are assigned to the same cluster, 
indicating that they belong to the same type.
LSH is particularly well-suited for this task as it approximates similarity across large datasets without requiring expensive pairwise comparisons, thus ensuring efficiency. 

We examine two LSH approaches: Euclidean LSH (ELSH)~\cite{datar2004locality,DBLP:books/cu/LeskovecRU14,spark_mllib_lsh} 
also known as $p$-stable or bucketed random projections, which is tailored to $\ell_2$ distance, and MinHash LSH~\cite{leskovec2020mining,spark_mllib_lsh,christen2020building}, used for Jaccard similarity over sets. 
% ELSH is based on Euclidean distance, while MinHash on Jaccard similarity~\cite{christen2020building}. 
Although they can be applied better to different cases~\cite{wang2017survey}, in our heterogeneous context, they are both efficient (\myref{\S}{sec:evaluation})~\cite{jafari2021survey}.

\mpara{ELSH.} 
ELSH~\cite{DBLP:books/cu/LeskovecRU14} is particularly suitable for our feature vectors, which combine label embeddings with binary property indicators. 
These vectors appear in a high-dimensional numeric space, where Euclidean distance works well. 
ELSH has two parameters~\cite{spark_mllib_lsh,apache_spark_lsh}: 
%\filledcircnum{A} 
a) the \textit{bucket length} $b>0$, which controls the width of the hash buckets and thus the granularity of similarity (larger bucket $\Rightarrow$ more collisions, higher recall but lower precision), and 
%\filledcircnum{B} 
b) the \textit{number of hash tables} $T \in \mathbb{N}$, which affects recall (higher $T$ $\Rightarrow$ more chances to collide) but increases runtime. 
% When increasing the number of hash tables lowers the false negatives and decreasing it improves the running performance. On the other hand, a larger bucket lowers the false negative rate.
ELSH is effective for real-valued embeddings, but requires parameter tuning to balance precision and recall~\cite{wang2017survey,jafari2021survey}.

\mpara{MinHash.}
MinHash LSH~\cite{leskovec2020mining}, approximates the Jaccard similarity between sets. 
It creates signatures for each set so the probability of two sets to collide in a hash function is equal to their Jaccard similarity. 
In our context, when nodes and edges are modeled simply as sets of properties, MinHash provides a simple similarity approach. 
As a parameter, MinHash only requires the number of hash tables $T$~\cite{spark_mllib_lsh,apache_spark_lsh}, which affects the trade-offs between recall and efficiency. 
MinHash performs well for sparse, set-like data, but it is less suitable when vectors include continuous components such as embeddings~\cite{wang2017survey,jafari2021survey}.

\mpara{Collision probabilities and parameter effects.} Each dataset requires different parameters to achieve the best possible clustering~\cite{njoku2024finding}, as it depends on the structural and statistical characteristics of the graph. 
The labels and the sparsity of properties affect the optimal bucket length and number of hash tables. 
For instance, denser graphs with many overlapping properties benefit from finer bucket lengths to better separate similar types, while sparser graphs require longer buckets to avoid mixing types. 
Similarly, the number of hash tables balances recall with computational efficiency.

Let $r_{\text{in}}$ be the distance between items of the same type and $r_{\text{out}}>r_{\text{in}}$ the distance between different types.  
For \textbf{ELSH}, the collision probability of two points with distance $d$ in one table is $p_b(d)$, a decreasing function of $d$ determined by the bucket length $b$~\cite{datar2004locality}. 
With $T$ independent hash tables combined under the OR rule, the probability that two vectors collide in at least one table is
$P_{b,T}(d) \;=\; 1-\bigl(1-p_b(d)\bigr)^T$ ~\cite{leskovec2020mining}, which decreases as $d$ grows and increases with $T$. 
Hence, decreasing $b$ (narrower buckets) or increasing $T$ increases selectivity and so improves the separation between
${P_{b,T}(r_{\text{in}})} $ and ${P_{b,T}(r_{\text{out}})}$,
leading to clusters with fewer mixed types (i.e., fewer false negatives)~\cite{gionis1999similarity,huang2015query}.

For \textbf{MinHash}, the collision probability for two sets $A,B$, is the probability that their signatures agree in one hash function equals their Jaccard similarity, $\Pr[h(A)=h(B)] = J(A,B) = \frac{|A\cap B|}{|A\cup B|}$~\cite{leskovec2020mining}.
With $T$ hash functions, the collision rate remains $J(A,B)$, and this estimation becomes more reliable as $T$ increases.
In practice, similar sets collide often, and dissimilar ones rarely, making MinHash a simple similarity approach.

% overal impact of parameter selection 

\mpara{Adaptive parameterization.} 
To take into account the heterogeneity of property graphs and avoid manual tuning of LSH, we introduce an adaptive strategy for selecting the key parameters $b$ and $T$, leveraging the characteristics of the dataset. 

Before clustering, we examine a small portion of the graph to infer key characteristics, like how sparse the dataset is or how many labels the dataset has. This helps us guide the parameter selection for  clustering. More specifically, we randomly sample 1\% of the graph, or at least 10k nodes (whichever is larger), and compute the Euclidean distances between the sampled elements and take their average as the distance scale $\mu$. 
This ensures that the bucket length is adapted to the actual distance of the data.
Then, we set the bucket width proportional to  $\mu$, as: $b_{\text{base}}= 1.2 \cdot \mu$. The factor $1.2$ avoids overfragmentation when the sample distance is small.

To also consider the number of labels, we refine $b$ using the number of distinct labels $L$ observed in the dataset:  $b =  b_{\text{base}}\cdot \alpha$, where $\alpha=0.8$ for $L \leq 3$, $\alpha=1.0$ for $4 \leq L \leq 10$, and $\alpha=1.5$ for $L > 10$. This heuristic reflects that graphs with few labels require tighter buckets to keep types distinct, while graphs with many labels benefit from wider buckets to avoid overfragmentation.  At this point, we prefer more separate types, as we are going to perform a merging step afterwards (\myref{\S}{sec:merging}). 

The number of hash functions is scaled according to the dataset size $N$ and then adjusted according to the number of labels $L$. Smaller graphs use fewer hash functions, while larger graphs demand more to distinguish the different types. Graphs with few labels need more hash functions to prevent mixed elements from being grouped, while graphs with many labels can separate from the label information and therefore need fewer functions.
The value of $T$ is heuristically determined by $T = b_{\text{base}} \cdot\max(5, \alpha \cdot \min(25, \log_{10} N))$, for nodes and   $T = b_{\text{base}} \cdot\max(3, \alpha \cdot \min(20, \log_{10} E))$ for edges.
% , inspired by previous work~\cite{}. 

These heuristics ensure a balance between precision and recall across graphs of varying size and label diversity, while avoiding both over-fragmentation and excessive merging. Regardless of the adaptive approach, users can always provide their own LSH parameters, that suit the respective dataset.

\mpara{Practical ranges.}
Empirically for our experiments, $\alpha\in[0.5,2]$ and $T\in[15,35]$ work well across
datasets; edges benefit from slightly smaller $\alpha\in[0.5,1.5]$, due to smaller vector representations, however, we explore the space of alternatives in (\myref{\S}{sec:evaluation}).
 
\mpara{Clustering result.} %After setting the parameters, we perform the clustering, using either the ELSH~\cite{DBLP:books/cu/LeskovecRU14} or MinHash~\cite{leskovec2020mining}. The result of  clustering is a collection of clusters of nodes and edges. 
A cluster is a group of nodes or edges that have similar structural characteristics, such as properties and labels. From these structural similarities (i.e., patterns \myref{Def.}{def:node_pattern},\myref{}{def:edge_pattern}), we define a type pattern from the cluster representative.
% which is the union of the properties and labels shared by its members.
Additionally, we maintain the associated instances of each cluster that represent this type. 
This way, we specify a type pattern and collect the set of instances assigned to it.
% The cluster also maintains the instances that represent this type. This way, each cluster specifies a type pattern and collects the set of instances assigned to it.

\mpara{Cluster representative.}
Each cluster $C$ (of nodes or edges) is summarized by a \emph{representative pattern} $\mathit{rep}(C) = (L,K,R)$, where:
\begin{itemize}[noitemsep, topsep=0pt, left=0pt]
    \item $L$ is the set of all labels that appear in at least one instance of $C$,
    \item $R =(L_s,L_t)$ for edges, is the set of all labels that appear in at least one instance on the source and target label sets, respectively, and
    \item $P$ is the set of properties observed across the instances of $C$.
\end{itemize}
We refer to the type patterns emerging form the clustering process, each representing a potential schema element, as \emph{candidate types}. We consider our candidate types (node or edges) $T$ as the cluster $C$ representative itself. These clusters form the basis for the type inference, as they capture the structure of the data, and they are used in the next steps.

\examplebox{
\begin{exmp}
Nodes ``Bob'' and ``John'' (\myref{Figure}{fig:pg-acm}) both have the label \texttt{Person} and have the same structure (\texttt{name}, \texttt{gender}, \texttt{bday}), so they are assigned together. The unlabeled node ``Alice'' has the same property vector, but it is not placed in the same cluster.
The two \texttt{Post} nodes are initially assigned in different clusters due to structural differences—one includes an \texttt{imgFile} property while the other a \texttt{content}.
\texttt{LIKES} edges are grouped together, as they consistently link a \texttt{Person} node to a \texttt{Post}.
%\HK{can you give something more than text visually?}
\end{exmp}
}

\subsection{Extracting Types}
\label{sec:merging}

After clustering, we refine the candidate node and edge types by merging clusters that correspond to the same schema type. 
We distinguish two cases. If a cluster has at least one label, we treat it as \emph{labeled}, otherwise as \emph{unlabeled}, so we can handle them differently and avoid mixing different types. An overview of the heuristics used is described in \myref{Algorithm}{alg:mergetypes}.

\mpara{Merging labeled clusters.}  
Node and edge clusters that have the same label(s) are merged directly, as they represent the same type, and preserve their information (\myref{Lemma}{lemma:monotonicity_nodes} and \myref{Lemma}{lemma:monotonicity_edges}).

\mpara{Merging unlabeled clusters.}  
For unlabeled clusters, we compare them against labeled ones using the Jaccard similarity of their property sets.  
Given two clusters $C_1, C_2$ with type patterns
$T_1=(\mathcal{L}_1,\mathcal{K}_1)$ and $T_2=(\mathcal{L}_2,\mathcal{K}_2)$, the similarity is
$
J(C_1, C_2) = \frac{|\,\mathcal{K}_1 \cap \mathcal{K}_2\,|}{|\,\mathcal{K}_1 \cup \mathcal{K}_2\,|}.
$
If $J(C_1,C_2) \geq \theta$ (we set $\theta=0.9$), they are merged into
$
T_M = (\mathcal{L}_1 \cup \mathcal{L}_2,\; \mathcal{K}_1 \cup \mathcal{K}_2).
$
We use a high similarity threshold to avoid over-merging; lowering $\theta$ would increase recall but mix types and will decrease precision.

\begin{lemma}[Monotonicity of merging node types]
\label{lemma:monotonicity_nodes}
Let $T_{N1}=(\mathcal{L}_1,\\\mathcal{K}_1)$,
$T_{N2}=(\mathcal{L}_2,\mathcal{K}_2)$, and
$T_{NM}=(\mathcal{L}_1\cup\mathcal{L}_2,\;\mathcal{K}_1\cup\mathcal{K}_2)$
be the merge of $T_{N1}$ and $T_{N2}$. Then,
$\mathcal{K}_i\subseteq \mathcal{K}_M$ and $\mathcal{L}_i\subseteq \mathcal{L}_M$ for $i\in\{1,2\}$.
Thus, no node property or node label is lost.
\end{lemma}
\begin{proof}
By construction:
$\mathcal{L}_M=\mathcal{L}_1\cup\mathcal{L}_2$,
$\mathcal{K}_{e,M}=\mathcal{K}_{e,1}\cup\mathcal{K}_{e,2}$.
\end{proof}

\begin{lemma}[Monotonicity of merging edge types]
\label{lemma:monotonicity_edges}
Let $T_{E1}=(\mathcal{L}_1, \\ \mathcal{K}_{e,1},\mathcal{R}_1)$,
$T_{E2}=(\mathcal{L}_2,\mathcal{K}_{e,2},\mathcal{R}_2)$ with
$\mathcal{R}_1=(\mathcal{L}_{s,1},\mathcal{L}_{t,1})$, and $\mathcal{R}_2=(\mathcal{L}_{s,2},\mathcal{L}_{t,2})$ and
$T_{EM}=\bigl(\ \mathcal{L}_1\cup\mathcal{L}_2,\ \mathcal{K}_{e,1}\cup\mathcal{K}_{e,2},\
(\mathcal{L}_{s,1}\cup\mathcal{L}_{s,2},\ \mathcal{L}_{t,1}\cup\mathcal{L}_{t,2})\ \bigr)$ be the merge of $T_{E1}$ and $T_{E2}$.
Then, $\mathcal{L}_i\subseteq\mathcal{L}_M$, $\mathcal{K}_{e,i}\subseteq\mathcal{K}_{e,M}$,
and $\mathcal{R}_1,\mathcal{R}_2\subseteq\mathcal{R}_M$.
% $i\in\{1,2\}$.
Thus, no label, property, or endpoint is lost.
\end{lemma}

\begin{proof}
By construction:
$\mathcal{L}_M=\mathcal{L}_1\cup\mathcal{L}_2$,
$\mathcal{K}_{e,M}=\mathcal{K}_{e,1}\cup\mathcal{K}_{e,2}$,
and $\mathcal{R}_M=\mathcal{R}_1 \cup\mathcal{R}_2=(\mathcal{L}_{s,1}\cup\mathcal{L}_{s,2},\ \mathcal{L}_{t,1}\cup\mathcal{L}_{t,2})$.
\end{proof}

\noindent This ensures that in the merging step properties and labels are not eliminated, only added. 
% It also guarantees that later steps \myref{\S}{sec:post_processing} have full information.

\mpara{Unmatched unlabeled clusters.}  
If an unlabeled cluster is not merged with labeled ones, we then merge it with other unlabeled clusters using the same Jaccard approach. Finally, if it remains unlabeled, we keep it as an \texttt{ABSTRACT} type (as in the PG-Schema~\cite{DBLP:journals/pacmmod/AnglesBD0GHLLMM23}). 

\mpara{Edges.}  
We merge edges only by label and get the set of source and target node types to define the connectivity $\rho_s(e_t)=(t_{\text{src}},t_{\text{tgt}})$. 
% for each edge type. 

\examplebox{
\begin{exmp}
The nodes ``Bob'', ``John'', and ``Alice'' (\myref{Figure}{fig:pg-acm}) all have identical property sets (\texttt{name}, \texttt{gender}, \texttt{bday}). The cluster containing ``Alice'' lacks a type label, but it has high Jaccard similarity with the cluster labeled as \texttt{Person}, which includes ``Bob'' and ``John''. Based on this similarity, we merge Alice's cluster with the \texttt{Person} cluster. In another case, the two \texttt{Post} nodes have different structures. However, since they both have the same label, ``\texttt{Post}'', we merge them.
\end{exmp}
}

\vspace{-4mm}
\begin{algorithm}[h]
\small
\SetAlgoLined

\SetKwIF{If}{ElseIf}{Else}
  {\textbf{if}}{\textbf{ then}}
  {\textbf{ else if}}{\textbf{ else}}
  {} 

\SetKwFor{ForEach}{\textbf{for each}}{\textbf{ do}}{}
\SetKwFor{For}{\textbf{for}}{\textbf{ do}}{}
\SetKwFor{While}{\textbf{while}}{\textbf{ do}}{}

\caption{Extracting and Merging Types}
\label{alg:mergetypes}

\KwIn{
Clusters $\mathcal{C} = \{C_1, C_2, \dots, C_m\}$ with labels $\lambda(C)$ and properties $\mathcal{P}_C$, 
Jaccard similarity threshold $\theta \in [0,1]$
}
\KwOut{Refined set of types $T$}

\SetKwFunction{Merge}{mergeClusters}
\SetKwFunction{Jac}{Jaccard}

$T \gets \emptyset$ \;

\ForEach{$C_i \in \mathcal{C}$}{
  \If{$\lambda(C_i) \neq \emptyset$}{
    \If{$\exists\, T_j \in T : \lambda(T_j) = \lambda(C_i)$}{
      \Merge{$C_i, T_j$} \;
    }\Else{
      add $C_i$ to $T$\;
    }
  }
}

\ForEach{$C_u \in \mathcal{C}$ with $\lambda(C_u) = \emptyset$}{
  $Cands \gets \{ C_l \in T \mid \lambda(C_l) \neq \emptyset \wedge \Jac(\mathcal{K}_{C_u}, \mathcal{K}_{C_l}) \ge \theta \}$ \;
  \If{$Cands \neq \emptyset$}{
    \Merge{$C_u, Cands$} \;
  }
}

\ForEach{pair $(C_u, C_v)$ of unlabeled clusters}{
  \If{$\Jac(\mathcal{K}_{C_u}, \mathcal{K}_{C_v}) \ge \theta$}{
    \Merge{$C_u, C_v$} \;
  }
}

\Return $T$
\end{algorithm}

\vspace{-6mm}

\subsection{Post Processing}
\label{sec:post_processing}

As mentioned before, we compute, optionally, more characteristics of the schema, such as whether a property is optional or mandatory, their data types and edge cardinalities. This helps in providing a precise schema, which supports validation processes, and can be described through 
% the STRICT graph schema representation, provided by 
the PG-schema~\cite{DBLP:journals/pacmmod/AnglesBD0GHLLMM23} grammar.

\mpara{Property constraints.}
Next, we identify mandatory and optional properties for each type from the merging step.
A property is characterized as mandatory for a given type $T$ if it appears in every instance of that type, otherwise, it is considered optional. 
Let $I_T = \{i_1, i_2, \dots, i_n\}$ be the instances of type $T$ (either nodes or edges), and let $p$ be a property observed in at least one instance. Let
$f_T(p) = \frac{|\{i \in I_T \mid p \in \mathcal{P}_i\}|}{|I_T|}
$ denote the frequency of property $p$ in type $T$.
A property $p$ is considered \emph{mandatory} for a type $T$, if $f_T(p) = 1$, else, $p$ is \emph{optional}, assigning the property constraint to $\pi_s(t, p)$.

\examplebox{
\begin{exmp}
Properties \texttt{name}, \texttt{gender}, 
and \texttt{bday} are mandatory for the type \texttt{Person} (\myref{Figure}{fig:pg-acm}), since all \texttt{Person} instances have them.
%The \texttt{Person} type (\myref{Figure}{fig:pg-acm}) has all properties present in all members of the type and so are mandatory.
In contrast, some \texttt{Post} instances have the \texttt{imgFile} property, while others don't, so we consider it as an optional property for the type \texttt{Post}.
\end{exmp}
}

\mpara{Property data types.}
To enrich the schema with property types, we analyze %a sample of a 
the data of the property $p$ with values $v_1,v_2, \dots, v_k$ from each type $T$. A priority-based inference is applied by checking the data type of each value $v$. First, for numeric types, $\text{if } v \in \mathbb{Z}$, then we consider $p$ to be an integer, $\text{if } v \in \mathbb{R} \setminus \mathbb{Z}$, then $p$ is a float,  $\text{if } v \in \{\texttt{true}, \texttt{false}\}$ then $p$ is Boolean, and finally, using regex for date/time ISO formats, defaulting to a string, assigning the property constraint to $\pi_s(t, p)$.
As these heuristics are expensive performance-wise, optionally we add a flag to infer this information, by sampling a small amount of data (e.g., 10\% of the properties, and at least 1000), which notably has minor differences \myref{\S}{sec:evaluation}, as we fallback to a string default.  
We leave for future work the identification of more detailed datatypes, such as enumerated types or bounded ranges.

\examplebox{
\begin{exmp}
For the \texttt{Person} type in \myref{Figure}{fig:pg-acm}, the properties
\texttt{name} and \texttt{gender} are assigned as strings, and \texttt{bday} is inferred as a date based on its format (e.g., \texttt{19/12/1999}).
\end{exmp}
}

\mpara{Cardinalities.}
To find the cardinalities $\mathcal{C}$ of edge types,
we query how many distinct targets each source node connects to, for each edge
type, and vice versa.
Specifically, let \(\mathcal{E}\) be the set of edges in the graph, and let \(\mathrm{src}(e)\) and \(\mathrm{tgt}(e)\) in $V$ denote the source and target node of an edge $e \in \mathcal{E}$, respectively.
For each edge type \(\rho\), we compute the maximum out-degree \(\max_{\mathrm{out}}(\rho) = \max_s |\{ t \mid (s \rightarrow t) \in \mathcal{E},\ \mathrm{type}(s \rightarrow t) = \rho \}|\) and the maximum in-degree \(\max_{\mathrm{in}}(\rho) = \max_t |\{ s \mid (s \rightarrow t) \in \mathcal{E},\ \mathrm{type}(s \rightarrow t) = \rho \}|\). 
The pair \((\max_{\mathrm{out}}(\rho), \max_{\mathrm{in}}(\rho))\) is then interpreted as follows: \((1,1)\) implies a $0$:$1$ relationship, \((>1,1)\) a $N$:$1$, \((1,>1)\) a $0$:$N$, and \((>1,>1)\) an $M$:$N$ cardinality. We cannot determine whether the source's lower bound is exactly 0 or 1, as we query only the edges. This requires to examine if all nodes are connected to the respective edge, which increases computational time; we leave this as future work. 
 This information shows the constraints in the graph, which can be leveraged in tasks such as data validation, consistency enforcement, or query optimization.

\examplebox{
\begin{exmp}
The \texttt{WORKS\_AT} edge (\myref{Figure}{fig:pg-acm}) connects the \texttt{Person} nodes to exactly one \texttt{Org.}, while an organization may have multiple employees. Thus, the inferred cardinality is \(N{:}1\).
Similarly, the \texttt{KNOWS} relationship connects people to other people, so the inferred cardinality is \(M{:}N\).
\end{exmp}

}

\subsection{Schema Serialization} 
 % \HK{should we call it Schema Serialization?}
From 
% all 
the previous steps, we assemble a  property graph schema $S_G = (V_s, E_s, \rho_s
% , \lambda_s, \pi_s, \mathcal{C}
)$ following \myref{Def.}{def:schema_graph}. This schema expresses the complete structure of the data. For interoperability, we export the schema in XSD \cite{DBLP:journals/corr/abs-1202-4532} and PG-Schema \cite{DBLP:journals/pacmmod/AnglesBD0GHLLMM23}, 
% formats, 
enabling easy integration into external tools.
% ~\cite{} 
% \sophia{pg2rdf}. 
%for validation \cite{DBLP:journals/corr/abs-1902-06427}, querying, or visualization. 
Even though PG-schema has not yet been expressed in a standard structured language~\cite{chantharojwong2025lics,ahmetaj2025common}, we generate both \texttt{LOOSE} and \texttt{STRICT} graph 
% example 
schema declarations for demonstration\footnote{
\href{https://github.com/sophisid/PG-HIVE/blob/master/schemadiscovery/pg_schema_output_loose.txt}{\textcolor{linkblue}{LOOSE}}, \href{https://github.com/sophisid/PG-HIVE/blob/master/schemadiscovery/pg_schema_output_strict.txt}{\textcolor{linkblue}{STRICT SCHEMA}}
 }.
 The \texttt{LOOSE} schema type can be used for flexible data insertions, allowing nodes and edges to deviate, and the \texttt{STRICT}, which also describes data types and schema constraints. \texttt{STRICT} schema demands a rigorous structure, which can be overwhelming for real datasets, which include noisy, incomplete, and inconsistent data.

\subsection{Incremental Step} 
\label{sec:incremental}
To make it possible to process large datasets on machines with limited memory, we introduce an \emph{incremental} approach. Instead of recomputing the schema from scratch, PG-HIVE can process new data insertions in small batches. Every new batch stream $Gs$ of nodes and edges is first transformed into vectors and clustered, just like in the initial pipeline. The resulting clusters are then merged with the existing schema (\myref{Algorithm}{alg:mergetypes}), ensuring that similar types are combined while avoiding redundant definitions. This merging step extends the schema with new patterns when necessary, while maintaining consistency with previously discovered types:

%\noindent So we can formally define  the merging of the schema as:

\mpara{Schema merging.}\label{def:schema-merging}  
Let $S_1 = (V_{s1}, E_{s1}, \rho_{s1} 
% \lambda_{s1}, \pi_{s1}, \mathcal{C}_{1}
)$ and  
$S_2 = (V_{s2}, E_{s2}, \rho_{s2}
% , \lambda_{s2}, \pi_{s2}, \mathcal{C}_{2}
)$  
be two property graph schemas.  The schema merging produces a new schema $S_{\text{merged}} = (V_m, E_m, \rho_m),$ 
% , \lambda_m, \pi_m, \mathcal{C}_m
such that:  
 $\forall G_1 \in G(S_1): G_1 \in G(S_{\text{merged}}), \quad  
  \forall G_2 \in G(S_2): G_2 \in G(S_{\text{merged}})$, 
  where $G(S)$ denotes the set of graphs conforming to schema $S$.  
  $S_{\text{merged}}$ is the least general schema that can be valid, so it maintains meaningful type definitions and structural coherence, while general enough to accommodate both $S_1$ and $S_2$.

\mpara{Merge rules.} Similarly with the \myref{Algorithm}{alg:mergetypes} and \myref{\S}{sec:merging}, we describe the merge rules as follows:
\begin{itemize}[leftmargin=*,topsep=1pt]
\item\textbf{Node types.} $V_m$ is constructed by unifying $V_{s1}$ and $V_{s2}$.  
Labeled types with the same label(s) are merged. Unlabeled types 
are merged firstly with labeled, then with unlabeled types if they have similar structure,   
otherwise, they are included as new \texttt{ABSTRACT} types.  

\item\textbf{Edge types.} $E_m$ is created by merging edge types with the same labels and updating their connectivity function $\rho_m(e)$.  

\item\textbf{Properties.} $\pi_m$ is defined as the union of properties from $\pi_{s1}, \pi_{s2}$.  
\end{itemize}

Thus, $S_{\text{merged}}$ generalizes both $S_1$ and $S_2$, avoids excessive generalization, and reflects the least general schema that ensures compatibility.

\mpara{Monotonicity.}
The correctness of the merging step in the incremental approach, is based on the \myref{Lemma}{lemma:monotonicity_nodes} and \myref{Lemma}{lemma:monotonicity_edges} which guarantees that when two types are merged, 
their labels, properties and relations are preserved due to their union. 
Therefore, no information is ever discarded on the schema, while every label and property present in $S_1$ or $S_2$ is also present in $S_{\text{merged}}$. 

Formally, let $S_i$ be the schema after processing batch $i$, and $S_{i+1}$ the schema after batch $i{+}1$. 
Then $S_i \sqsubseteq S_{i+1}$, meaning that $S_{i+1}$ is a generalization of $S_i$, extending it with new labels and properties, but without removing existing ones. 
Thus, the sequence $(S_1, S_2, \dots, S_n)$ forms a monotone chain of schemas, 
where each preserves and extends the previous one. 

By updating the schema incrementally, PG-HIVE adapts as the graph grows. Newly observed structures are immediately incorporated and constraints can be refined on demand. This avoids expensive full recomputation while keeping the schema accurate and responsive to evolving data. In practice, this makes the method suitable for dynamic environments where updates are frequent. Handling updates and deletions is left for future work, but the current incremental strategy already reduces computational cost significantly while ensuring that the schema remains up to date.

\subsection{Theoretical Guarantees}
\label{sec:guarantees}
Following, we discuss the guarantees of the inferred schema.
% in terms of \emph{scalability}, \emph{completeness}, \emph{soundness}, and \emph{consistency} of the inferred schema.

\mpara{Efficiency.} PG-HIVE remains scalable as it avoids pairwise comparisons by using the LSH. The comparisons would need a computational time of $\mathcal{O}(N^2 D)$. LSH reduces this to 
$\mathcal{O}(N T D)$ (ELSH) or $\mathcal{O}(N T)$ (MinHash), where $N$ is the number of elements, 
% to be clustered,
$T \ll N$ is the number of tables, and $D$ the dimension of the embedding. 

\mpara{Type completeness.}  
Let $G=(V,E)$ be the input property graph and $S_G=(V_s,E_s,\rho_s)$ the schema inferred by PG-HIVE.  
For every node $v \in V$ with label $\ell$ and set of properties $\mathcal{P}_v$, PG-HIVE guarantees that there exists a type $t \in V_s$ such that 
$\ell \in \lambda_s(t) \quad \text{and} \quad \mathcal{P}_v \subseteq \pi_s(t)$.  
This means that no label or property of the graph is lost. This is ensured since we take the union of each cluster, rather than the intersection, so information is preserved.

\mpara{Property constraints.} 
PG-HIVE distinguishes mandatory and optional properties in a sound way.
For a type $T$ with instances $I_T$, a property $p$ is mandatory only if it appears in \emph{all} instances of $T$ (\myref{\S}{sec:post_processing}),
otherwise, it is optional. Hence, every property marked as mandatory is indeed present in every instance of the type.

\mpara{Data type inference.}  
For each property, we examine its observed values and assign the most specific compatible type following a simple hierarchy:  
\texttt{integer}, \texttt{float}, \texttt{date/time}  and  \texttt{string}.  
As a result, all values of a property are consistent with the inferred type, even though the type may be a generalization as string.
% because we assign string as default.

\mpara{Cardinalities.}
For every edge type, PG-HIVE computes the maximum in- and out-degree observed in the data $
(\max_{\mathrm{out}}(\rho), \max_{\mathrm{in}}(\rho))$.
These values serve as upper bounds to the respected cardinalities. For example, if a relationship is inferred as $(0,N)$, then, for each source node attached to this edge 
% in the data is connected to
has more than one target.

\mpara{Schema merging.}  
When merging two schemas $S_1$ and $S_2$, the resulting schema $S_{\text{merged}}$ is guaranteed to be general enough to cover both schemas (\myref{\S}{sec:incremental}). Labeled clusters with the same label are merged, unlabeled clusters are merged only when structurally similar, and otherwise kept as abstract types. This ensures that the merged schema is a general schema covering both inputs.

\mpara{Incrementality.}  
When data arrives in batches, PG-HIVE extends the schema without invalidating what has already been discovered. Formally, if $S_i$ is the schema after batch $i$, and $S_{i+1}$ after the next batch, then $S_i \sqsubseteq S_{i+1}$ \myref{\S}{sec:incremental},
meaning that $S_{i+1}$ generalizes $S_i$. So, no previously valid instance is excluded, and the schema evolves monotonically as the graph grows.

In short, PG-HIVE guarantees that (i) all labels and properties observed in the data are preserved, (ii) mandatory and optional properties are classified correctly, (iii) inferred datatypes are always compatible with observed values; if fully computed, (iv) edge cardinalities reflect sound upper bounds, and (v) schemas evolve monotonically when updated incrementally. Together, these guarantees ensure that PG-HIVE produces consistent and faithful schemas without losing information from the underlying graph.
%\vspace{-2mm}

\mpara{Time complexity.}
\label{sec:complexity}
\textit{Static module.} We examine the three main components: preprocessing, clustering and type extraction. Preprocessing transforms each element
%instance 
into a vector representation using binary property indicators and fixed-dimensional Word2Vec embeddings with the time complexity $\mathcal{O}(N \cdot (P + D))$, where $N$ is the number of data elements (nodes or edges), $P$ the number of properties and $D$ the embedding size.
In LSH, each vector is projected into $T_n$ number hash tables, and computing each projection requires a small cost of $\mathcal{O}(D)$. Thus the clustering requires $\mathcal{O}(N \cdot T_n \cdot D)$ time. Finally, the worst-case scenario of the merging step where the LSH clusters lack any labels, requires computing the Jaccard similarity between all pairs of clusters, has a time complexity of $\mathcal{O}(C_n^2)$, where $C_n$ is the number of clusters produced.
Combining these steps the total complexity of PG-HIVE is: $\mathcal{O}(N \cdot (P + T_n \cdot D)) + \mathcal{O}(C_n^2)$. In practice, $P$, $T_n$ and $D$ are constants, and typically much smaller than $N$, so the complexity is dominated by $\mathcal{O}(N)$ and the cost of the merging step $\mathcal{O}(C_n^2)$, resulting in a 
% final
time complexity of $\mathcal{O}(N + C_n^2)$.

\textit{Incremental module.} 
In the incremental approach, each new batch of $B$ data elements (nodes or edges) passes through the same three steps: preprocessing, clustering, and merging. 
Similarly, the preprocessing requires $\mathcal{O}(B \cdot (P + D))$, and the clustering  $\mathcal{O}(B \cdot T_n \cdot D)$. 
In the merging step, let $C_b$ denote the number of clusters in the current batch. 
These clusters are compared against the existing $C_n$ clusters of the schema, which results in a cost of $\mathcal{O}(C_b \cdot C_n)$. 
So the total complexity per batch is
$ \mathcal{O}(B \cdot (P + T_n \cdot D)) + \mathcal{O}(C_b \cdot C_n).$ Again, $P$, $T_n$ and $D$ are constants, and much smaller than $N$, so the final complexity is $\mathcal{O}(B + (C_b \cdot C_n))$.

Compared to the static case,
% where the merging requires $\mathcal{O}(C_n^2)$, 
the incremental design reduces the workload in smaller updates. 
Since $B \ll N$ and $C_b \ll C_n$, the incremental approach scales due to the batch size, reducing the overhead while ensuring that the schema remains up to date.

\begin{table}[t]
\setlength{\tabcolsep}{3pt}
\centering
\scriptsize
\caption{Dataset statistics.}
\label{tab:datasets_characteristics}

\vspace{-3mm}
\resizebox{\columnwidth}{!}{%
\begin{tabular}{|l|c|c|c|c|c|c|c|c|c|}
\hline
Dataset & Nodes & Edges & 
\shortstack{Node\\Types} &
\shortstack{Edge\\Types} &
\shortstack{Node\\Labels} &
\shortstack{Edge\\Labels} &
\shortstack{Node\\Pat.} &
\shortstack{Edge\\Pat.} &
\shortstack{Real \ \\Synth.}  \\
\hline
POLE  & 61,521    & 105,840  &  11 &  17 & 11 & 16  & 17 & 16 & \textcolor{orng}{S}\\
MB6    & 486,267    & 961,571    & 4  & 5  & 10 & 3  & 52 & 4 & \textcolor{orng}{S} \\
HET.IO & 47,031     & 2,250,197  & 11 & 24 & 12 & 24  & 14 & 35 & \textcolor{frstgreen}{R} \\
FIB25  & 802,473    & 1,625,428    & 4  & 5  & 10 & 3  & 31 & 4 & \textcolor{orng}{S}\\
ICIJ   & 2,016,523  & 3,339,267  & 5 & 14 & 6 & 14 & 208 & 42 &\textcolor{frstgreen}{R}\\
CORD19 & 5,485,296 & 5,720,776 & 16  & 16  & 16  & 16  & 89 & 6  & \textcolor{frstgreen}{R}\\
LDBC   & 3,181,724  & 12,505,476 & 7  & 17 & 8  & 15 & 9  & 15 & \textcolor{orng}{S} \\
IYP    & 44,539,999 & 251,432,812 & 86 & 25 & 33 & 25 & 1210 & 790 & \textcolor{frstgreen}{R}\\
\hline
\end{tabular}
}
\vspace{-4mm}
\end{table}

\section{Evaluation}
\label{sec:evaluation}

This section presents our evaluation methodology and results.

\mpara{Setup.}
All experiments were conducted on a 4-node Spark Standalone cluster  (Ubuntu 20.04.2 LTS, 38 cores/machine, SSDs, Gigabit Ethernet), using Spark 3.4.1, Scala 2.12.10, and Neo4j 4.4.0. All configurations and code are publicly available for reproducibility (\myref{\S}{sec:artifacts}).

\mpara{Datasets.}
Dataset characteristics are summarized in \myref{Table}{tab:datasets_characteristics}. We used four synthetic and four real datasets for the evaluation (indicated as \textcolor{orng}{S} and \textcolor{frstgreen}{R}, respectively). We also include the individual structural patterns (\myref{Def.}{def:node_pattern},\myref{}{def:edge_pattern}) of the various datasets, to highlight the variations in the structure of instances. 
%POLE~\cite{neo4j_pole}, MB6~\cite{pdf-neuroprint,takemura2017connectome},
%HET.IO~\cite{himmelstein2017hetionet}, 
%FIB25~\cite{pdf-neuroprint,takemura2015synaptic}, 
%ICIJ~\cite{icij_bahamas}
%LDBC~\cite{DBLP:conf/sigmod/ErlingALCGPPB15,DBLP:journals/pvldb/SzarnyasWSSBWZB22}, CORD19~\cite{covidgraph2021} and IYP~\cite{fontugne2024wisdom}. 

POLE~\cite{neo4j_pole} is a small benchmark with crime investigation data.
HET.IO~\cite{himmelstein2017hetionet} integrates biomedical entities, such as genes, diseases, and drugs. FIB25~\cite{takemura2015synaptic} and MB6~\cite{takemura2017connectome} model connectome data and correspond to the mushroom body
and, respectively, medulla neural networks in the fruit fly brain. ICIJ~\cite{icij_bahamas} contains information on offshore entities, linked to major leaks such as the Panama Papers. The LDBC~\cite{DBLP:conf/sigmod/ErlingALCGPPB15,DBLP:journals/pvldb/SzarnyasWSSBWZB22} simulates a large-scale social network. CORD19~\cite{covidgraph2021} integrates genotype, disease, and bibliographical data. Finally, IYP~\cite{fontugne2024wisdom} describes networking and internet measurements, from integrated datasets.
% Dataset characteristics are summarized in \myref{Table}{tab:datasets_characteristics}. The table also includes the individual structural patterns (\myref{Def.}{def:node_pattern},\myref{}{def:edge_pattern}) of the various datasets, to highlight the variations in structure of instances. 

% \input{tables/datasets}

For all datasets, we already have the schema ground truth (node and edge types), yet none of them include property data types and constraints. Some of these datasets are already used for evaluating PG schema discovery~\cite{lbath2021schema,DBLP:conf/edbt/BonifatiDM22}, or PG-related tasks~\cite{bonifati2024dtgraph,cazzaro2025zograscope, 10577554,giakatos2025pythia}. However, to the best of our knowledge, this is the first time that all of them are used together in a single benchmark, significantly increasing the total number of datasets compared to previous works.
%Each dataset has diverse properties, as well as multi-labeled nodes, and is used to evaluate the discovery of the property graph schema, also in other works~\cite{lbath2021schema,DBLP:conf/edbt/BonifatiDM22}, or used in PG-related tasks~\cite{bonifati2024dtgraph,cazzaro2025zograscope, 10577554}. 

\mpara{Noise injection}. To evaluate the resilience of the various approaches in harder cases, we introduce noise to the datasets. 
We injected noise by randomly removing 0\%–40\% of node/edge properties, testing three label availability scenarios: 100\% (all labels retained), 50\% (half retained), and 0\% (no labels).
This approach enables us to examine how different pattern inconsistencies and poor label annotation affect the quality of schema discovery.

\mpara{Baselines.}
We compare our approach using ELSH~\cite{DBLP:books/cu/LeskovecRU14} and MinHash~\cite{leskovec2020mining},
against GMMSchema~\cite{DBLP:conf/edbt/BonifatiDM22} (GMM), and SchemI~\cite{lbath2021schema}, the only competitors available in the domain.

\mpara{Evaluation metrics.} We evaluate the \textbf{quality }of the resulting schema and the \textbf{efficiency} of the corresponding algorithms. In all approaches, a discovered type is represented as a cluster accompanied by a set of properties.
Each discovered cluster may contain data of multiple types, and the label(s) associated with a cluster are not known beforehand — they are inferred from the composition of the cluster itself after the clustering process. For evaluation purposes, we assign the label(s) to each cluster, namely the most frequent (majority) actual label(s) among the nodes or edges it contains (e.g., if most elements have label \texttt{Person} or \texttt{\{Person \& Student\}}). 
To evaluate the quality of the generated clusters, we use the \textbf{majority-based F1-Score (F1*-Score)}, where the correctness of a node/edge placement is determined based on whether its actual type matches the majority label(s) of its cluster. This approach has been adopted in previous works to assess clustering performance~\cite{mahmood2023optimizing}. This reflects how well the discovered clusters assign the types to the same groups, as well as how we represent the schema types, especially in noisy and semi-labeled datasets.

Beyond clustering quality, we also focus on evaluating the correctness of inferred data type attributes. We compare our sample-based inference against a full scan of the dataset. Formally, for each property $p$, let $D_p$ denote the set of all values of $p$ (by scanning the full dataset), and $S_p \subseteq D_p$ the sampled values. Let $f(p)$ be the datatype inference function. 
We define \textbf{the sampling error} for $p$ as, $\mathrm{error}(p) = \frac{1}{|S_p|} \sum_{v \in S_p} \mathbf{1}\big(f(v) \neq f(D_p)\big)$, where $f(D_p)$  is the inferred datatype when considering all values of $p$, 
and $\mathbf{1}(\cdot)$ equals 1 if the sample-based inference disagrees with the full-scan inference, 0 otherwise.

\usetikzlibrary{positioning}
\definecolor{elsh}{RGB}{255,165,0}
\definecolor{gmm}{RGB}{160,48,64}
\definecolor{minhash}{RGB}{34,139,34}
\definecolor{schemi}{RGB}{10,50,188} 
\tikzset{
  colELSH/.style={draw=elsh, line width=1pt},
  colGMM/.style={draw=gmm, line width=1pt},
  colMinHash/.style={draw=minhash, line width=1pt},
  colSchemI/.style={draw=schemi, line width=1pt},
}

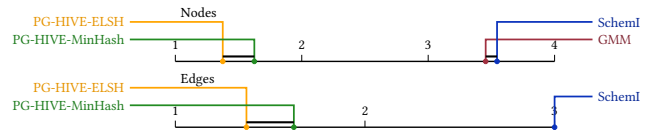
\begin{figure}[t]
\centering

\begin{minipage}{\columnwidth}
\centering
\resizebox{\columnwidth}{!}{%

\begin{tikzpicture}[x=7.6cm,y=2.2cm]
\def\gap{0.60}  \def\L{-0.12} \def\R{1.10}
\def\tickup{0.06} \def\leadup{0.28}

% ---------- Edges (k=3) ----------
\def\y{0*\gap}
\draw[line width=0.8pt] (0,\y) -- (1,\y);
\foreach \x/\r in {0/1,0.5/2,1/3}{
  \draw (\x,\y) -- (\x,\y+\tickup);
  \node[font=\small, above] at (\x,\y+\tickup) {\r};
}
\node[font=\small, anchor=west] at (0,\y+\leadup+0.15) {Edges};

\def\xEL{0.1875} \def\xMH{0.3125} \def\xSC{1.0000}

\draw[colELSH]   (\L,\y+\leadup+0.08) -- (\xEL,\y+\leadup+0.08) -- (\xEL,\y);
\filldraw[elsh]   (\xEL,\y) circle (1.3pt);
\node[elsh!90!black, font=\small, anchor=east]  at (\L,\y+\leadup+0.08) {PG-HIVE-ELSH};

\draw[colMinHash](\L,\y+\leadup-0.08) -- (\xMH,\y+\leadup-0.08) -- (\xMH,\y);
\filldraw[minhash] (\xMH,\y) circle (1.3pt);
\node[minhash!90!black, font=\small, anchor=east] at (\L,\y+\leadup-0.08) {PG-HIVE-MinHash};

\draw[colSchemI] (\xSC,\y) -- (\xSC,\y+\leadup) -- (\R,\y+\leadup);
\filldraw[schemi] (\xSC,\y) circle (1.3pt);
\node[schemi!80!black, font=\small, anchor=west] at (\R,\y+\leadup) {SchemI};

% non-significant group
\draw[line width=1.2pt] (\xEL,\y+0.045) -- (\xMH,\y+0.045);

% ---------- Nodes (k=4) ----------
\def\y{1*\gap}
\draw[line width=0.8pt] (0,\y) -- (1,\y);
\foreach \x/\r in {0/1,0.3333/2,0.6667/3,1/4}{
  \draw (\x,\y) -- (\x,\y+\tickup);
  \node[font=\small, above] at (\x,\y+\tickup) {\r};
}
\node[font=\small, anchor=west] at (0,\y+\leadup+0.17) {Nodes};

\def\xEL{0.1250} \def\xMH{0.2083} \def\xRS{0.8333} \def\dx{0.015}

\draw[colELSH]   (\L,\y+\leadup+0.08) -- (\xEL,\y+\leadup+0.08) -- (\xEL,\y);
\filldraw[elsh]   (\xEL,\y) circle (1.3pt);
\node[elsh!90!black, font=\small, anchor=east]  at (\L,\y+\leadup+0.08) {PG-HIVE-ELSH};

\draw[colMinHash](\L,\y+\leadup-0.08) -- (\xMH,\y+\leadup-0.08) -- (\xMH,\y);
\filldraw[minhash] (\xMH,\y) circle (1.3pt);
\node[minhash!90!black, font=\small, anchor=east] at (\L,\y+\leadup-0.08) {PG-HIVE-MinHash};

\draw[colGMM]    (\xRS-\dx,\y) -- (\xRS-\dx,\y+\leadup-0.08) -- (\R,\y+\leadup-0.08);
\filldraw[gmm]    (\xRS-\dx,\y) circle (1.3pt);
\node[gmm!80!black, font=\small, anchor=west]    at (\R,\y+\leadup-0.08) {GMM};

\draw[colSchemI] (\xRS+\dx,\y) -- (\xRS+\dx,\y+\leadup+0.08) -- (\R,\y+\leadup+0.08);
\filldraw[schemi] (\xRS+\dx,\y) circle (1.3pt);
\node[schemi!80!black, font=\small, anchor=west] at (\R,\y+\leadup+0.08) {SchemI};

% non-significant groups
\draw[line width=1.2pt] (\xEL,\y+0.045) -- (\xMH,\y+0.045);
\draw[line width=1.2pt] (0.8183,\y+0.045) -- (0.8483,\y+0.045);

\end{tikzpicture}}

\end{minipage}

\caption{Statistical significance analysis of F1-scores across datasets for nodes (top) and edges (bottom) --GMM does not produce edge types. 
% \HK{I don't know the test, give me one line on how to interpret the numbers/figures I see}
}
% \vspace{-10mm}
\label{fig:signif-f1}
\vspace{-4mm}
\end{figure}
\begin{figure*}[!t]
\resizebox{\linewidth}{!}{%
\input{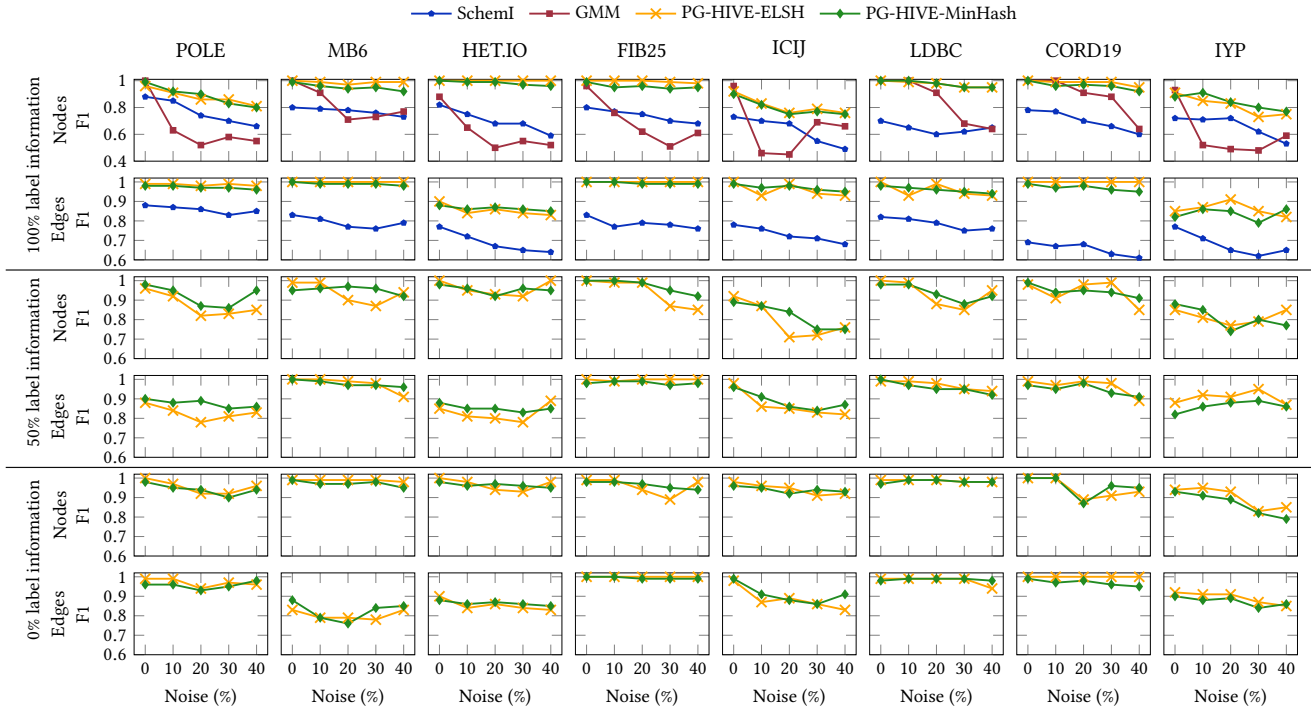}
}
\vspace{-7mm}
\caption{F1 scores across all noise levels (0-40\%) and label availability (0-50-100\%).}
\label{fig:f1score}
\vspace{-3mm}
\end{figure*}

\subsection{Results}

\mpara{Statistical significance.} First, we focus on the effectiveness and present the F1*-scores across all 40 test cases (8 datasets x 5 noise levels) under 100\% label availability. \textit{Note that GMM and SchemI are able to work only under fully labeled datasets (100\% label availability).}
Using the Nemenyi test~\cite{nemenyi1963distribution, herbold2020autorank} we show in \myref{Figure}{fig:signif-f1} 
which pairs of algorithms differ significantly. 
Each point of the \myref{Figure}{fig:signif-f1} indicates the average rank of the method over all cases of noise, with lower ranks indicating better average performance. As shown,  PG-HIVE-ELSH and PG-HIVE-MinHash form a group with no major difference between them, while both significantly outperform GMM and SchemI for nodes, and SchemI for edges --GMM cannot discover edge types. This confirms that PG-HIVE is statistically better across all test cases. 

%OLD
% \mpara{Quality of the discovered schema as noise increases.}
% Then we focus in detail on the F1*-scores across different levels of noise. The \myref{Figure}{fig:f1score} shows the F1*-scores as noise increases (0-40\%), under three label-availability (100\%, 50\%, 0\%). 
% Across all datasets and test cases, PG-HIVE maintains high accuracy with scores above 0.9 even in high levels of noise and no label availability. 
% This indicates that PG-HIVE is reliable, when properties are incomplete in addition with poor label annotation.
% In contrast, GMM and SchemI are more sensitive to absence of properties and labels. In the case that both operate (100\% label availability) SchemI is consistently lower (0.6-0.8) than PG-HIVE. GMM starts with high F1* score but as noise exceeds 20\%, the F1* drops below 0.6 additionally it cannot discover edge types.

% In the corner case where no noise is ingested, PG-HIVE and GMM both accurately discover the correct number of node types, whereas they are both superior to SchemI. However, as noise increases, our approach still accurately infers the number of types and correctly groups the corresponding elements, whereas the performance of both GMM and SchemI significantly drops. Regarding edge types, only SchemI is able to discover them with far inferior performance than PG-HIVE, even in cases without noise.
% \sophia{ \HK{this should go to the evaluation results, here we still describe setup/competitors/datasets}

\mpara{Quality of the discovered schema as noise increases.}
Then we focus in detail on the F1*-scores across different levels of noise. \myref{Figure}{fig:f1score} shows the F1*-scores as noise increases (0-40\%), under three label availability scenarios (100\%, 50\%, 0\%). 

Overall, across all datasets and test cases, PG-HIVE maintains high accuracy with scores above 0.9 even in high levels of noise and no label availability. 
This happens due to PG-HIVE's hybrid approach, leveraging LSH to approximate structural similarities effectively, which remains robust as noise misleads the clustering. Additionally, the merging step further refines clusters, resulting in accurate results.
%This indicates that PG-HIVE is reliable,  even when properties are incomplete in addition with poor label annotation.
In the corner case where no noise is ingested (100\% label availability and 0\% noise), PG-HIVE and GMMSchema both accurately discover the correct number of node types (F1*-score approx. 1.0) across datasets, whereas they are both superior to SchemI. However, as noise increases, PG-HIVE's approximation and adaptive merging retain accurate inference and correctly group the corresponding elements by focusing on structural patterns. 

In contrast, GMMSchema and SchemI are sensitive to noise and missing labeling information. For the node type inference, in 100\% label availability, GMMSchema starts with a high F1*-score (1.0), as the Gaussian Mixture Model clustering performs well on clean data, but as noise exceeds 20\%, the variety in property distributions causes misclustering, dropping the F1*-score below 0.6. SchemI consistently has a lower F1*-score (0.6-0.8) than PG-HIVE. For the edge types, only SchemI is able to detect them with a significantly inferior performance than PG-HIVE.

% \HK{what you said applies to node types, right? if yes, then update previous statements to refer to node types and add: For the edge types, only SchemI is able to detect them with a significantly inferior performance than PG-HIVE.} 

% On the other hand, the performance of both GMM and SchemI significantly drops. Regarding edge types, only SchemI is able to discover them with far inferior performance (0.5-0.6) than PG-HIVE, even in cases without noise.

% \HK{I see repetition of the same things in this paragraph. make Uniform and remove repetition}

With incomplete labeling information, only PG-HIVE is able to produce results, and as such, we don't see the SchemI and GMMSchema lines for 50\% and 0\% labeling information. Although type inference becomes straightforward when label availability is 100\%, with 50\% label availability, slight misgroupings may occur due to similar structures in different types.
%, using the Jaccard similarity-based grouping to refine clustering, which maintains a relatively high accuracy
In the 0\% label availability case, inference becomes more challenging, as types with identical structures are merged based on structural similarity using both LSH and the second merging step, potentially reducing precision but still enabling robust discovery with PG-HIVE. 

Comparing the two PG-HIVE variations, there is no significant difference (see \myref{Figure}{fig:signif-f1}), in terms of F1*-scores. They both maintain high accuracy (above 0.9) under different noise levels and label availability, as their LSH-based clustering and second-step merging are efficient enough to capture structural similarities. The slight differences, if any, occur due to the adaptive parameterization \myref{\S}{sec:clustering}, which might overlook some outliers; however this is tackled with our adaptive parameterization.% this by having a lower and upper bound.  

% \HK{do we see any difference between ELSH and MinHash?}

% \HK{are edges more difficult than nodes?}
% Comparing the inference between nodes and edges, 
Comparing node and edge types, inference of node types can be more challenging. Nodes may appear unlabeled in the datasets, along with different structural variations. While edges usually have fewer structural variations, extracting their types relies on their labeling information and the concatenation of the different (if such) source and target labels. 
Additionally, the corresponding internal LSH representation and the merging step give an advantage to edge discovery. Nodes are represented as a Word2Vec along with the one-hot encoded binary vector, while edges have 3 Word2Vec and a binary vector. This increased expressivity of the Word2Vec allows a better separation in edge clustering, when labels are not identical, resulting in higher F1* scores (above 0.9) even under noise. And finally, the second merging step groups together edges with the same labels, if separated before due to structural differences.
% while adding the various nodes types to the set of source and target nodes. 
% Thus, extracting the underlying structure of nodes can be more challenging.
\\

\mpara{PG-HIVE across datasets.} Each dataset has diverse properties, as well as multi-labeled nodes, which increase the perplexity of type inference. Simpler or homogeneous datasets, such as POLE, MB6, FIB25, and HET.IO have a flat structure and are easier to infer the schema, even in the absence of label annotation. In contrast, ICIJ and IYP, which are integrated datasets and structurally heterogeneous, pose additional challenges due to structural variability. In these cases, PG-HIVE slightly declines as noise grows, while baselines degrade greatly. LDBC and CORD19, while larger in size, have fewer structural differences, and so remain simple to discover their types. 
Additionally, MB6 and FIB25 obtain nodes annotated with multiple labels,  complicating type inference due to overlapping semantic roles, which can lead to misgroupings if we rely only on the structural similarity, under high levels of noise. Similarly, LDBC, HET.IO, ICIJ and IYP include one extra label that may add semantic information (e.g., HET.IO has assigned to all its nodes an extra \texttt{HetionetNode} label, for integration scenarios). These multi-labeling scenarios require careful consideration of label and structural information in PG-HIVE’s clustering, particularly in heterogeneous datasets where baseline methods struggle to adapt.

\input{tikzz/execution-time}
\input{tikzz/heatmap}

\mpara{Efficiency.} \myref{Figure}{fig:exec_time} shows execution times across all datasets and noise levels. The execution time includes the preprocessing, clustering, and type extraction. 
We can observe that PG-HIVE (both ELSH and MinHash variants) has better execution times, by up to 1.95x on average than SchemI in all cases. PG-HIVE still maintains a comparable efficiency to the GMMSchema method, which only retrieves node types. 
PG-HIVE is able to process small and less complicated datasets (POLE, MB6, HET.IO, FIB25) in less than half a minute,  
bigger and more versatile datasets (LDBC, ICIJ, CORD19), in 3-7 minutes,  while 
% bigger and 
more complicated datasets (IYP) need 15 minutes. 

We can notice that, in our approach, noise does not affect the computational time, also justified by \myref{\S}{sec:complexity}. However, this is not the case for GMMSchema. As the noise increases, the number of clusters grows, resulting in higher computational cost, and in most cases, at 40\% noise, PG-HIVE is more efficient, despite the fact that GMMSchema only discovers node types. PG-HIVE's efficiency can, in addition, be theoretically justified by LSH's approximate hashing, which scales linearly with data size (O(N)), unlike GMM's quadratic complexity in noisy scenarios, making PG-HIVE suitable for large and unstructured graphs.

\mpara{Adaptive parameters.}
In \myref{Figure}{fig:heatmap}, we ran experiments in the ELSH approach with different parameters $(T, a)$, comparing against our adaptive approach (marked in red). 
The results of the F1*-score indicate that in most datasets, the adaptive approach is very close to the best-performing setting (yellow indicates higher F1*, thus blue lower). This happens due to the upper bound in the heuristic-based estimation, which narrows down the search space. When performance drops, it is due to the second merging step (\myref{\S}{sec:merging}), not parameter miscalculation. However, as this relies on approximate heuristics, there might be cases that we cannot find the optimal parameter, but still we can identify one with high accuracy. IYP is such a case; however, although we could not identify the optimal parameters, we still retained high accuracy.
% Validating the adaptive tuning's effectiveness without the need for manual intervention. 
% \HK{????} \sophia{oti den xreiazetai na  baloun to parameter tou LSH hardcoded, mporei na fygei i teleytaia protasi}

In general, lowering $b$ values over-separates patterns; thus, we have a higher F1*-score, while increasing $b$ and $T$ values, merges distinct patterns, lowering F1*-score. This is also in line with our theoretical expectations (\myref{\S}{sec:clustering}).

\begin{figure}

\pgfplotsset{
  compat=1.18,
  twopanel/.style={
    width=0.49\textwidth,
    height=0.21\textwidth,
    grid=both,
    xlabel={Iteration}, ylabel={Exec. Time (s)},
    xtick={1,...,10},
    tick label style={font=\small},
    label style={font=\small},
  }
}

\tikzset{
  leg1/.style={color={rgb,255:red,255;green,165;blue,0},   thin, mark=x,          mark size=2.5pt},
  leg2/.style={color={rgb,255:red,200;green,48;blue,64},   thin, mark=x,    mark size=2.5pt},
  leg3/.style={color={rgb,255:red,128;green,0;blue,128},   thin, mark=triangle*,   mark size=1pt},
  leg4/.style={color={rgb,255:red,34;green,139;blue,34},   thin, mark=otimes*,  mark size=1pt},
  leg5/.style={color={rgb,255:red,51;green,153;blue,255},  thin, mark=triangle*,  mark size=1pt},
  leg6/.style={color={rgb,255:red,20;green,5;blue,200},   thin, mark=x,    mark size=2pt},
  leg7/.style={color={rgb,255:red,210;green,105;blue,30},  thin, mark=star,       mark size=1pt},
  leg8/.style={color={rgb,255:red,0;green,0;blue,0},       thin, mark=star,          mark size=1pt},
}

\centering

\begin{tikzpicture}

\begin{groupplot}[
  group style={
    group size=1 by 3,
    horizontal sep=0.42cm,
    vertical sep= 0.8cm,
    group name=gp
  },
  title style={font=\footnotesize, yshift=-5pt},
  twopanel,
]

\nextgroupplot[title={PG-HIVE-ELSH}, xlabel={},
  ymode=log,
  log basis y=10,
  ymin=0, ymax=100,
  ytick={0.5,1,2,5,10,20,50,100},
  ymajorgrids, yminorgrids,
  log ticks with fixed point]
  \addplot[leg1] coordinates {(1,0.52) (2,0.49) (3,0.47) (4,0.50) (5,0.53) (6,0.48) (7,0.51) (8,0.54) (9,0.50) (10,0.49)}; % POLE
  \addplot[leg2] coordinates {(1,2.08) (2,1.92) (3,2.10) (4,1.85) (5,1.89) (6,1.82) (7,1.85) (8,1.80) (9,1.88) (10,1.85)}; % MB6 
  \addplot[leg3] coordinates {(1,0.46) (2,0.50) (3,0.51) (4,0.48) (5,0.52) (6,0.49) (7,0.53) (8,0.47) (9,0.50) (10,0.51)}; % HET.IO
  \addplot[leg4] coordinates {(1,2.05) (2,1.87) (3,2.10) (4,2.02) (5,1.95) (6,2.18) (7,2.00) (8,2.14) (9,2.07) (10,1.92)}; % FIB25
  \addplot[leg5] coordinates {(1,31.8) (2,29.9) (3,30.4) (4,30.4) (5,30.8) (6,28.9) (7,27.3) (8,27.6) (9,32.8) (10,33.4)}; % ICIJ 
  \addplot[leg6] coordinates {(1,9.0) (2,8.2) (3,8.8) (4,9.5) (5,8.5) (6,9.6) (7,8.6) (8,9.0) (9,9.2) (10,8.7)}; % LDBC
  \addplot[leg7] coordinates {(1,39) (2,41) (3,43) (4,41) (5,44) (6,42) (7,47) (8,49) (9,49) (10,43)}; % CORD19 
  \addplot[leg8] coordinates {(1,85) (2,80) (3,75) (4,78) (5,81) (6,76) (7,80) (8,77) (9,79) (10,73)}; % IYP

\nextgroupplot[
  title={PG-HIVE-MinHash},
  xlabel={},
  ymode=log,
  log basis y=10,
  ymin=0, ymax=100,
  ytick={0.5,1,2,5,10,20,50,100},
  ymajorgrids, yminorgrids,
  log ticks with fixed point]
  \addplot[leg1] coordinates {(1,0.52) (2,0.49) (3,0.47) (4,0.50) (5,0.53) (6,0.48) (7,0.51) (8,0.54) (9,0.50) (10,0.49)}; % POLE
  \addplot[leg2] coordinates {(1,2.1) (2,2.05) (3,2.10) (4,2.15) (5,1.89) (6,1.92) (7,1.95) (8,1.90) (9,1.95) (10,1.90)}; % MB6
  \addplot[leg3] coordinates {(1,0.46) (2,0.50) (3,0.51) (4,0.48) (5,0.52)(6,0.49) (7,0.53) (8,0.47) (9,0.50) (10,0.51)}; % HET.IO
  \addplot[leg4] coordinates {(1,2.05) (2,1.87) (3,2.10) (4,2.02) (5,1.95) (6,2.18) (7,2.00) (8,2.14) (9,2.07) (10,1.92)}; % FIB25
  \addplot[leg5] coordinates {(1,33.6) (2,32.8) (3,29.3) (4,29.2) (5,29.9) (6,29.8) (7,30.1) (8,34.5) (9,35.7) (10,30.2)}; % ICIJ
  \addplot[leg6] coordinates {(1,8.1) (2,8.4) (3,9.1) (4,9.8) (5,9.0) (6,9.9) (7,8.8) (8,9.2) (9,9.5) (10,8.9)}; % LDBC
  \addplot[leg7] coordinates {(1,50) (2,49) (3,45) (4,45) (5,45) (6,43) (7,42) (8,42) (9,40) (10,44)}; % CORD19
  \addplot[leg8] coordinates {(1,89) (2,82) (3,77) (4,81) (5,84) (6,78) (7,83) (8,79) (9,82) (10,75)}; % IYP

% \nextgroupplot[title={DiscoPG}, xlabel={},
%   ymode=log,
%   log basis y=10,
%   ymin=0.15, ymax=100,
%   ytick={0,0.5,1,2,5,10,20,50,100},
%   ymajorgrids, yminorgrids,
%   log ticks with fixed point]
%   \addplot[leg1] coordinates {(1,0.15) (2,0.16) (3,0.17) (4,0.20) (5,0.28) (6,0.31) (7,0.43) (8,0.60) (9,0.80) (10,0.90)}; % POLE ok
%   \addplot[leg2] coordinates {(1,0.31) (2,0.42) (3,0.46) (4,0.42) (5,0.59) (6,0.74) (7,1.24) (8,1.41) (9,1.58) (10,1.5)}; % MB6 ok
%   \addplot[leg3] coordinates {(1,0.51) (2,0.62) (3,0.76) (4,0.82) (5,0.89) (6,0.94) (7,1.24) (8,1.38) (9,1.482) (10,1.60)}; % HET.IO ok
%   \addplot[leg4] coordinates {(1,1.10) (2,1.11) (3,1.14) (4,2.39) (5,2.27) (6,2.54) (7,3.05) (8,3.28) (9,3.279) (10,3.5)}; % FIB25 ok
%   \addplot[leg5] coordinates {(1,25.1) (2,27.13) (3,28.12) (4,28.16) (5,28.180) (6,30.193) (7,32.195) (8,32.2) (9,35.198) (10,40.195)}; % ICIJ ok
%   \addplot[leg6] coordinates {(1,2.12) (2,2.14) (3,2.23) (4,2.25) (5,2.30) (6,2.34) (7,2.39) (8,2.51) (9,2.53) (10,2.58)}; % LDBC ok
%   \addplot[leg7] coordinates {(1,20.1) (2,21.22) (3,21.05) (4,21.10) (5,21.32) (6,22.47) (7,23.21) (8,23.56) (9,24.35) (10,24.52)}; % CORD19 ok
%   \addplot[leg8] coordinates {(1,50.10) (2,55.41) (3,57.13) (4,60.018) (5,65.25) (6,65.30) (7,65.34) (8,66.71) (9,65.28) (10,67.90)}; % IYP ok

\end{groupplot}

\node[anchor=north, xshift=-8mm, yshift=-30mm]
  at ($(gp c1r1.south)!0.5!(gp c2r1.south)$) {%
  \begin{tikzpicture}[baseline, font=\scriptsize]
    \matrix[column sep=2mm]{
      \draw[leg1] plot coordinates {(0,0) (0.4,0) (0.8,0)} node[right,xshift=2mm]{POLE}; &
      \draw[leg2] plot coordinates {(0,0) (0.4,0) (0.8,0)} node[right,xshift=2mm]{MB6}; &
      \draw[leg3] plot coordinates {(0,0) (0.4,0) (0.8,0)} node[right,xshift=2mm]{HET.IO}; &
      \draw[leg4] plot coordinates {(0,0) (0.4,0) (0.8,0)} node[right,xshift=2mm]{FIB25}; &\\
      \draw[leg5] plot coordinates {(0,0) (0.4,0) (0.8,0)} node[right,xshift=2mm]{ICIJ}; &
      \draw[leg6] plot coordinates {(0,0) (0.4,0) (0.8,0)} node[right,xshift=2mm]{LDBC}; &
      \draw[leg7] plot coordinates {(0,0) (0.4,0) (0.8,0)} node[right,xshift=2mm]{CORD19}; &
      \draw[leg8] plot coordinates {(0,0) (0.4,0) (0.8,0)} node[right,xshift=2mm]{IYP}; \\
    };
  \end{tikzpicture}%
};

\end{tikzpicture}

\vspace{-4mm}

\caption{Incremental execution time per iteration.}
\label{fig:incr}
\vspace{-2mm}
\end{figure}
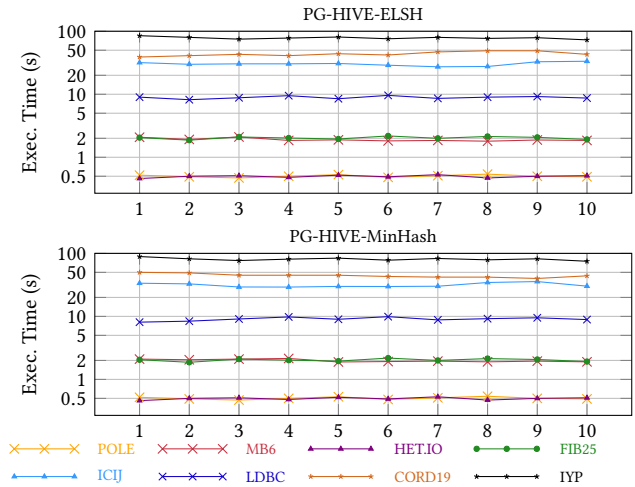
\mpara{Incremental.}
To evaluate the incremental approach, we randomly separate the graph into 10 batches and run both of our approaches. 
In \myref{Figure}{fig:incr},  we show the running time in seconds that each approach needs to process each batch of size $B$.
% Notably all approaches have minor differences as the input data is equal in size. The computational time is based mostly on the number of elements we have, and the candidate number of types, which is much smaller than the batch size(\myref{\S}{sec:complexity}). \sophia{
The incremental approach suits to streaming data, by avoiding full recomputation. The consistent running times justify the incremental design's efficiency, as it processes only new data and merges with the existing schema $(O(B + C_b * C_n))$, with $C_b,C_n \ll B$.
% }
% \HK{you beg to ask you what happens in bigger sizes here}
% Specifically, DiscoPG can a slight increase in the computational time, due to the increase in the number of clusters. As more data is loaded, new types can be discovered, resulting in a higher computational time in DiscoPG. 
% Thus, in our approach, the computational time is based mostly on the number of elements we have, and the number of types.

 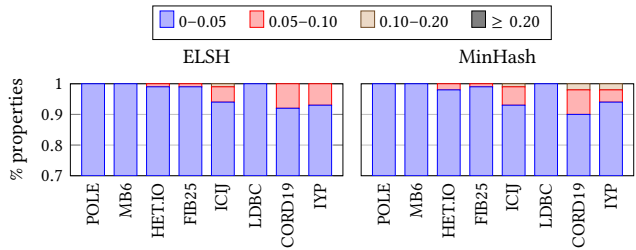
\begin{figure}[t]
\centering
\resizebox{\columnwidth}{!}{%

\begin{tikzpicture}
\begin{groupplot}[
  group style={group size=2 by 1, horizontal sep=8pt},
  width=0.9\linewidth, height=3cm,
  ymin=0.7, ymax=1,
  ybar stacked,
  x=0.5cm,
  symbolic x coords={POLE,MB6,HET.IO,FIB25,ICIJ,LDBC,CORD19,IYP},
  xtick=data,
  xticklabel style={font=\small, rotate=90},
  yticklabel style={font=\small},
  ymajorgrids
]

% ---------- Left axis ----------
\nextgroupplot[
  title={ELSH},
  ylabel={\% properties},
  legend to name=DeltaBinsLegend,
  legend columns=5,
  legend style={font=\small, /tikz/every even column/.style={column sep=8pt}}
]
\addplot coordinates {(POLE,1) (MB6,1) (HET.IO,0.99) (FIB25,0.99) (ICIJ,0.94) (LDBC,1) (CORD19,0.92) (IYP,0.93)};
\addplot coordinates {(POLE,0) (MB6,0) (HET.IO,0.01) (FIB25,0.01) (ICIJ,0.05) (LDBC,0) (CORD19,0.08) (IYP,0.07)};
\addplot coordinates {(POLE,0) (MB6,0) (HET.IO,0) (FIB25,0) (ICIJ,0.01) (LDBC,0) (CORD19,0) (IYP,0.0)};
\addplot coordinates {(POLE,0) (MB6,0) (HET.IO,0) (FIB25,0) (ICIJ,0) (LDBC,0) (CORD19,0) (IYP,0.0)};
\legend{{0–0.05},{0.05–0.10},{0.10–0.20},{$\geq 0.20$}} % 

% ---------- Right axis ----------
\nextgroupplot[
  title={MinHash},
  yticklabels={},
  ymajorticks=false,
  ylabel={}
]
\addplot coordinates {(POLE,1) (MB6,1) (HET.IO,0.98) (FIB25,0.99) (ICIJ,0.93) (LDBC,1) (CORD19,0.90) (IYP,0.94)};
\addplot coordinates {(POLE,0) (MB6,0) (HET.IO,0.02) (FIB25,0.01) (ICIJ,0.06) (LDBC,0) (CORD19,0.08) (IYP,0.04)};
\addplot coordinates {(POLE,0) (MB6,0) (HET.IO,0) (FIB25,0) (ICIJ,0.01) (LDBC,0) (CORD19,0.02) (IYP,0.02)};
\addplot coordinates {(POLE,0) (MB6,0.00) (HET.IO,0) (FIB25,0) (ICIJ,0) (LDBC,0) (CORD19,0) (IYP,0.0)};

\end{groupplot}

\node at ($(group c1r1.north east)!0.5!(group c2r1.north west)$) [above=16pt]
  {\pgfplotslegendfromname{DeltaBinsLegend}}; 
\end{tikzpicture}
}
\vspace{-6mm}
\caption{Distribution of data type inference errors using sampling, across datasets for ELSH (left) and MinHash (right).}
\label{fig:delta}
\vspace{-4mm}
\end{figure}
\mpara{Schema constraints.}
The inference of optional/mandatory attributes and cardinalities has previously been explained in ~\myref{\S}{sec:guarantees}. In this direction, we focus next on evaluating the quality of the identified datatypes. In our approach, we infer the datatypes based on sampling (\myref{\S}{sec:post_processing}). As such, we compare our approach with the dominant types determined using a full scan, in terms of sampling error.  The results are shown in \myref{Figure}{fig:delta}, grouped into bins (0–0.05, 0.05–0.10, 0.10–0.20, and $\geq$0.20), and normalized by the number of properties per dataset. 

As shown, most properties fall into the lowest error range, demonstrating consistent datatype inference across datasets for both clustering approaches.
The few outliers occur on bigger datasets with heterogeneous properties like ICIJ, CORD19, and IYP, where the small sample does not reflect all the properties available in the data. Misinterpretations of data are usually of the following types: assigning \texttt{DOUBLE} rather than \texttt{INTEGER}, or assigning \texttt{DATE}, but the full scan might observe outliers and assign the property as \texttt{STRING}. Nevertheless, this problem can be easily mitigated by increasing the sampling percentage and is dataset dependent. We leave this exploration for future work.
% \HK{move this to where you define the evaluation metrics} Formally, for each property $p$, let $D_p$ denote the set of all values of $p$ (by scanning the full dataset), 
% and $S_p \subseteq D_p$ the sampled values. Let $f(p)$ be the datatype inference function. 
% We define the sampling error for $p$ as, $\mathrm{error}(p) = \frac{1}{|S_p|} \sum_{v \in S_p} \mathbf{1}\big(f(v) \neq f(D_p)\big)$, 
% where $f(D_p)$  is the inferred datatype when considering all values of $p$, 
% and $\mathbf{1}(\cdot)$ equals 1 if the sample-based inference disagrees with the full-scan inference, 0 otherwise.

\mpara{Summary.} Our extensive evaluation on both real and synthetic datasets demonstrates that PG-HIVE consistently outperforms existing schema discovery approaches,  even under high levels of noise and in the absence of label information.
In terms of quality, PG-HIVE maintains high accuracy, even under challenging scenarios with 0\% label information, where GMMSchema and SchemI do not work at all, and 40\% noise, where GMMSchema and SchemI underperform. Specifically, PG-HIVE achieves up to 65\% higher F1*-score for node types and up to 40\% for edge type discovery compared to existing approaches. The statistical significance analysis (\myref{Figure}{fig:signif-f1}) confirms that both PG-HIVE approaches (ELSH and MinHash) are superior to prior methods across all datasets. 

In terms of efficiency, PG-HIVE has slightly slower execution times than GMMSchema, which can be justified as GMMSchema only infers node types, while PG-HIVE discovers both node and edge types, as well as constraings. Nevertheless, PG-HIVE remains efficient, achieving on average 1.95x faster execution compared to SchemI.
Moreover, our incremental approach demonstrates an efficient running time on each dataset.
Our adaptive parameter selection strategy leverages heuristics (\myref{\S}{sec:clustering}), which limit the exploration space while achieving near-optimal parameter selection.
Finally, datatype inference shows a small sampling error, demonstrating that the sampling-based method can be both reliable and efficient even for large datasets.  

% \sophia{
Overall, PG-HIVE provides a balance between accuracy and efficiency, outperforming the current state-of-the-art methods. This comprehensive evaluation justifies PG-HIVE's advancements, as it addresses real-world cases like multi-labeled and heterogeneous datasets, confirming its importance for property graph management.
% }

\vspace{-4mm}
\section{Conclusion and Future Work}
\label{sec:conclusion}

In this paper, we present PG-HIVE, a hybrid incremental schema discovery approach for property graphs. Our approach discovers node types, edge types, and infers schema characteristics, without any prior schema information and with high levels of noise. In addition, it can process large datasets in smaller batches incrementally. Unlike previous methods, it does not rely only on predefined labels; instead, it groups nodes and edges based on both structure and semantics. Our experiments show that it outperforms baseline approaches
and infers the schema successfully even with semi-annotated data.

As future work, we intend to investigate more challenging scenarios that further complicate schema discovery. Specifically, we aim to: a) handle cases where no label information is available and data is extremely sparse, b) detect types that share identical type patterns but lack distinguishing labels, and c) support integration scenarios when label semantics are not consistent (e.g., labels in different languages). 
% \sophia{
To address these challenges,  
we plan to enhance PG-HIVE to handle the label variations by integrating large language models (LLMs), to semantically align labels across datasets, without relying on exact string matches. Additionally, for enumerations and value semantics, we should leverage the property values and their semantic interpretation, along with additional schema constraints. 
% }
% should leverage the property values and their semantics, along with additional schema constraints. 
% \HK{rewrite to match the style of the beginning of the paragraph}

%Furthermore, we identify the need for benchmark dataset(s), given a detailed PG-schema covering both types and constraints. This will enable a more systematic evaluation, creating fair comparisons with other approaches, and show insights about the effectiveness of schema discovery methods.

 \section*{Artifacts}
 \label{sec:artifacts}
All configurations and code are publicly available for reproducibility here: \textcolor{linkblue}{\url{https://github.com/sophisid/PG-HIVE}}. 
The datasets used and their sources are listed as:
POLE~\cite{neo4j_pole} \href{https://github.com/neo4j-graph-examples/pole}{\textcolor{linkblue}{Github repo}}, 
MB6~\cite{takemura2017connectome,pdf-neuroprint} \href{https://github.com/sophisid/PG-HIVE/tree/master/datasets/MB6}{\textcolor{linkblue}{CSV Dataset}}, 
HET.IO~\cite{himmelstein2017hetionet} \href{https://github.com/hetio/hetionet}{\textcolor{linkblue}{Github repo}}, 
FIB25~\cite{takemura2015synaptic,pdf-neuroprint} \href{https://github.com/sophisid/PG-HIVE/tree/master/datasets/FIB25}{\textcolor{linkblue}{CSV Dataset}}, 
ICIJ~\cite{icij_bahamas} \href{https://github.com/ICIJ/offshoreleaks-data-packages}{\textcolor{linkblue}{Github repo}}, 
LDBC~\cite{DBLP:conf/sigmod/ErlingALCGPPB15,DBLP:journals/pvldb/SzarnyasWSSBWZB22} \href{https://github.com/sophisid/PG-HIVE/tree/master/datasets/LDBC}{\textcolor{linkblue}{CSV Dataset}}, 
CORD-19~\cite{covidgraph2021} \href{https://github.com/allenai/cord19}{\textcolor{linkblue}{Github repo}}, 
IYP~\cite{fontugne2024wisdom} \href{https://github.com/InternetHealthReport/internet-yellow-pages}{\textcolor{linkblue}{Github repo}}.

\bibliographystyle{splncs04}
\bibliography{bibliography}
\end{document}